%% file: main.tex
\documentclass[lettersize,journal]{IEEEtran}
\usepackage{amsmath,amsfonts}
\usepackage{algorithmic}
\usepackage{algorithm}
\usepackage{array}
\usepackage[caption=false,font=normalsize,labelfont=sf,textfont=sf]{subfig}
\usepackage{textcomp}
\usepackage{stfloats}
\usepackage{url}
\usepackage{verbatim}
\usepackage{graphicx}
\usepackage{cite}
\input{Format/packages}
\input{Format/CustomCommands} 
\input{Format/Glossary} 

\begin{document}

\title{\fontsize{24}{20}\selectfont Deep Broadcast Feedback Codes} 

\author{Jacqueline Malayter,~\IEEEmembership{Student Member,~IEEE}, Yingyao Zhou,~\IEEEmembership{Student Member,~IEEE,} Natasha Devroye,~\IEEEmembership{Fellow,~IEEE}, Chih-Chun Wang,~\IEEEmembership{Fellow,~IEEE,} Christopher Brinton,~\IEEEmembership{Senior Member,~IEEE,} David J. Love,~\IEEEmembership{Fellow,~IEEE}
\thanks{Jacqueline Malayter, Chih-Chun Wang, Christopher Brinton, and David J. Love are with the Elmore Family School of Electrical Engineering,
Purdue University, West Lafayette 47907, USA (emails: malayter@purdue.edu, chihw@purdue.edu, cgb@purdue.edu, djlove@purdue.edu).}
\thanks{Yingyao Zhou and Natasha Devroye are with the Electrical and Computer Engineering Department at the University of Illinois Chicago, Chicago, 60607, USA (emails: yzhou238@uic.edu, devroye@uic.edu).}
\thanks{\textit{(Jacqueline Malayter and Yingyao Zhou are co-first authors) (Corresponding author: Jacqueline Malayter)}}

}



\maketitle

\begin{abstract}
Recent advances in deep learning for wireless communications have renewed interest in channel output feedback codes. In the additive white Gaussian broadcast channel with feedback (AWGN-BC-F), feedback can expand the channel capacity region beyond that of the no-feedback case, but linear analytical codes perform poorly with even small amounts of feedback noise. Deep learning enables the design of nonlinear feedback codes that are more resilient to feedback noise. We extend single-user learned feedback codes for the AWGN channel to the broadcast setting, and compare their performance with existing analytical codes, as well as a newly proposed analytical scheme inspired by the learned schemes. Our results show that, for a fixed code rate, learned codes outperform analytical codes at the same blocklength by using power-efficient nonlinear structures and are more robust to feedback noise. 
Analytical codes scale more easily to larger blocklengths with perfect feedback and surpass learned codes at higher SNRs. 
\end{abstract}


\section{Introduction}\label{sec:Introduction}
\input{Introduction}
\section{Problem Setup}\label{sec:ProblemSetup}

\input{ProblemSetup}

\section{Analytical Codes}\label{sec:Analytical}
\input{Analytical}


\section{Main Results: Three Neural Code Constructions}\label{sec:Neural}
\input{Neural}


\section{Conclusion and Future Directions}\label{sec:Conclusion}
\input{Conclusion}

\bibliographystyle{IEEEtran}
\bibliography{main}

\appendix
\input{bclmaxsumrate}

\end{document}

%% file: Format/packages.tex
\usepackage{amsfonts}
\usepackage{amsmath}
\usepackage{amssymb}
\usepackage{amsthm}
\usepackage{mathtools}
\usepackage{xcolor}
\usepackage{glossaries}
\usepackage{stfloats}
\usepackage{url}
\usepackage{physics}
\usepackage{tikz}
\usepackage{mathdots}
\usepackage{yhmath}
\usepackage{cancel}
\usepackage{color}
\usepackage{siunitx}
\usepackage{array}
\usepackage{multirow}
\usepackage{gensymb}
\usepackage{tabularx}
\usepackage{extarrows}
\usepackage{booktabs}
\usetikzlibrary{fadings}
\usetikzlibrary{patterns}
\usetikzlibrary{shadows.blur}
\usetikzlibrary{shapes}

\usepackage{bm}

\usepackage{algorithm}
\usepackage{algorithmic}

\usepackage{titlesec}

\makeatletter
\newcommand{\arabicsubsubsection}{\arabic{subsubsection}}
\makeatother

\newcounter{subsubsubsection}[subsubsection]

\renewcommand\thesubsubsubsection{%
  \arabicsubsubsection\alph{subsubsubsection})%
}

\titleformat{\subsubsubsection}
  {\normalfont\itshape}
  {\thesubsubsubsection}
  {1em}
  {}

\titlespacing*{\subsubsubsection}
  {0pt}{3.25ex plus 1ex minus .2ex}{1.5ex plus .2ex}

\titleclass{\subsubsubsection}{straight}[\subsubsection]

\newtheorem{remark}{Remark}

%% file: Format/CustomCommands.tex
\newif\ifdraftmode
\draftmodefalse

\newcommand{\ifdraft}[2]{%
  \ifdraftmode
    #1%
  \else
    #2%
  \fi
}

\newcommand{\scriptW}{\mathcal{W}}
\newcommand{\scriptN}{\mathcal{N}}
\newcommand{\scriptX}{\mathcal{X}}

\newcommand{\R}{\mathbb{R}}
\newcommand{\F}{\mathbb{F}}

\newcommand{\nvect}{\mathbf{n}}
\newcommand{\xvect}{\mathbf{x}}
\newcommand{\yvect}{\mathbf{y}}
\newcommand{\zvect}{\mathbf{z}}
\newcommand{\evect}{\mathbf{e}}
\newcommand{\cvect}{\mathbf{c}}
\newcommand{\qvect}{\mathbf{q}}
\newcommand{\gvect}{\mathbf{g}}

\newcommand{\svect}{\mathbf{s}}

\newcommand{\Thetavect}{\bm{\Theta}}
\newcommand{\epsvect}{\bm{\epsilon}}
\newcommand{\avect}{\bm{\alpha}}
\newcommand{\rhovect}{\bm{\rho}}

\newcommand{\snrb}{\text{SNR}_b}

\newcommand{\eye}{\mathbf{I}}
\newcommand{\Fmat}{\mathbf{F}}

\newcommand{\Cmat}{\mathbf{C}}
\newcommand{\Rmat}{\mathbf{R}}

\newcommand{\paren}[1]{\left(#1\right)} 

\newcommand{\card}[1]{\text{card}\left(#1\right)}

\newcommand{\Prob}[1]{\mathbb{P}\left(#1\right)}
\newcommand{\expect}[1]{\mathbb{E}\left(#1\right)}

\newcommand{\lightcode}{\textsc{LightCode}}
\newcommand{\lightBCsep}{{\textsc{LightCode-BC-Sep}}}
\newcommand{\lightBC}{{\textsc{LightCode-BC}}}
\newcommand{\rpcBC}{{\textsc{RPC-BC}}}

\newtheorem{lemma}{Lemma}


%% file: Format/Glossary.tex
\newacronym{ol}{OL}{Ozarow-Leung}
\newacronym{eol}{EOL}{extended Ozarow-Leung}
\newacronym{lqg}{LQG}{linear quadratic Gaussian}
\newacronym{cl}{CL}{Chance-Love}
\newacronym{bcl}{BMCL}{broadcast Malayter-Chance-Love}

\newacronym{bler}{BLER}{block error rate}
\newacronym{awgnbcf}{AWGN-BC-F}{additive white Gaussian broadcast channel with channel output feedback}
\newacronym{awgnbc}{AWGN-BC}{additive white Gaussian broadcast channel}
\newacronym{awgn}{AWGN}{additive white Gaussian noise}
\newacronym{awgnf}{AWGN-F}{AWGN channel with output feedback}
\newacronym{su}{SU}{single-user}
\newacronym{dlecfc}{DL-ECFC}{deep-learned error-correcting feedback codes}
\newacronym{ecc}{ECC}{error correction code}

\newacronym{snr}{SNR}{signal-to-noise ratio}
\newacronym{pam}{PAM}{pulse amplitude modulation}
\newacronym{tdd}{TDD}{time division duplexing}
\newacronym{m2m}{M2M}{machine-to-machine}
\newacronym{iot}{IoT}{Internet-of-Things}
\newacronym{mac}{MAC}{multiple-access channel}

\newacronym{mmse}{MMSE}{minimum mean squared error}
\newacronym{lmmse}{LMMSE}{linear minimum mean squared error}
\newacronym{dare}{DARE}{discrete algebraic Riccati equation}

\newacronym{rnn}{RNN}{recurrent neural network}
\newacronym{gru}{GRU}{gated recurrent unit}
\newacronym{fe}{FE}{feature extractor}
\newacronym{mlp}{MLP}{multilayer perceptron}
\newacronym{cce}{CCE}{categorical cross-entropy}
\newacronym{ai}{AI}{artifical intelligence}

%% file: Introduction.tex
\IEEEPARstart{T}{he} evolution from fifth-generation New Radio (5G-NR) to sixth-generation (6G) networks will be characterized by data-driven designs that address the challenges of growing communication network complexity. \Gls{ai}-based applications create increasingly heterogeneous networks, warranting more flexible architectures. In 6G, \gls{ai}-native designs will be critical to enable intelligent, adaptive, and data-driven communication systems~\cite{brinton2025key,chowdhury6G}. The rise of autonomous systems and distributed learning applications will further drive a surge in \gls{m2m} and \gls{iot} communications, necessitating ultra-reliable and low-latency spectrum sharing methods~\cite{chowdhury6G}. For these applications, data packets are often short due to latency constraints and the nature of the data itself~\cite{durisiToward}. Consequently, 6G must leverage \gls{ai}, particularly deep learning, to design networks capable of meeting the demands of data-intensive \gls{ai} applications.


 \Glspl{ecc} will be crucial in 6G networks to ensure reliable data transmission, achieving low probabilities of error at high code rates to maximize throughput. Wireless access points supporting massive connectivity may use shorter packet sizes to reduce transmission latency~\cite{brinton2025key}. Deep learning can further enhance \gls{ecc} design in 6G by enabling efficient and reliable codes at short blocklengths. Broadcast channel codes are of particular interest, as they encode messages of multiple users into a single codeword, allowing simultaneous and reliable communication while improving spectral efficiency.

Cover first analyzed broadcast channels in~\cite{CoverBroadcast}, where a single source communicates simultaneously with multiple receivers. Although the capacity region of a general broadcast channel remains an open problem, it has been characterized for the \gls{awgn} case. In particular, Cover showed that the achievable rate region of the degraded \gls{awgn} broadcast channel extends beyond that obtained by simple time-sharing between two users~\cite{CoverBroadcast}. By contrast, in the symmetric \gls{awgn} case, where receivers have the same \gls{snr}, broadcast coding provides no advantage over time-sharing. However, the addition of channel output feedback \textit{changes} the achievable rate region of broadcast channels. While feedback does not increase the capacity of the single-user \gls{awgnf} channel~\cite{shannon1956zero}, it can enlarge the capacity region of the \gls{awgnbcf} when feedback is available from each receiver~\cite{dueck,ozarowAchievable,bhakaranGBCF}.

The enlarged capacity region of \gls{awgnbcf} codes makes them attractive for 6G applications. Typically, codes for the \gls{awgnbcf} have  been linear due to their simplicity, but linearity limits the design of optimal codes~\cite{ahmad2015concatenated}. Ozarow and Leung were the first to demonstrate the enlarged capacity region enabled by perfect feedback in the broadcast channel~\cite{ozarowAchievable}. They generalized the seminal \gls{lmmse}-based Schalkwijk-Kailath (SK) scheme from the \gls{awgnf} to the \gls{awgnbcf}, achieving doubly exponential error decay with blocklength~\cite{schalkwijk1966coding,ozarowAchievable}. Kramer later generalized \gls{lmmse}-based feedback codes to more than two users~\cite{Kramer2002}. Other schemes based on linear control theory further expand the capacity region beyond Ozarow's  scheme~\cite{elia2004bode,ardestanizadeh2012lqg,wu2005gaussian}. Moreover, Belhadj Amor \textit{et al.}~\cite{amorDuality} showed that linear quadratic Gaussian (LQG)-based schemes are sum-rate optimal among all linear-feedback codes for the symmetric single-antenna \gls{awgnbcf} with equal channel gains. In Gastpar \textit{et al.}~\cite{Gastpar}, the \gls{awgnbcf} is analyzed for the case of correlated noise, where a novel coding scheme based on successive noise cancellation is proposed. Other works have analyzed the discrete memoryless broadcast channel with feedback, such as in~\cite{WuRatelimited,shayevitz}, but we focus on the \gls{awgnbcf} in this work. 

Unfortunately, purely \textit{linear} feedback codes for the \gls{awgnbcf} fail to achieve positive rates with imperfect feedback, and their performance degrades quickly with increasing feedback noise power \cite{KimGaussNF,chance2011concatenated,ahmad2015concatenated}. Nonetheless, the addition of noisy feedback in the \gls{awgnbcf} can still extend the capacity region beyond that without feedback, if the coding schemes are not purely linear \cite{RaviCapacityGBC}. In the single user \gls{awgnf}, Chance \textit{et al.} leverage non-linearity by proposing a concatenated coding scheme based on concatenating an SNR-maximizing linear code, resulting in improved error exponent bounds over linear feedback schemes. In Farthofer \textit{et al.}, the authors propose a scheme utilizing quantization at the receiver input and the feedback channel input, demonstrating positive rates achievable rates with noisy feedback \cite{Farthofer2014}.  In another line of work by Mishra \textit{et. al} \cite{dynamicMishra}, dynamic programming is used to derive a closed-form optimal scheme for the \gls{awgnf} with noisy feedback. In \cite{vasal2021dynamic}, it is shown a dynamic program also exists to solve the capacity expression for the \gls{mac} channel. For the \gls{awgnbcf} with noisy feedback, Ahmad \textit{et al.}~\cite{ahmad2015concatenated} proposed a concatenated scheme that uses \gls{snr}-maximizing linear codes as the inner code, achieving rates beyond the no-feedback capacity region. On the other hand, recent advances in deep learning have enabled solutions to seemingly intractable wireless problems and opened new directions for designing high-performance non-linear codes. For example, in the \gls{su} \gls{awgnf}, \gls{dlecfc}~\cite{kim2018deepcode, safavi2021deep, mashhadi2021drf, shao2023attentioncode, ozfatura2022all, kim2023robust, ankireddy2024lightcode, ozfatura2023feedback, chahine2022inventing, lai2024variable} have been proposed for short blocklengths, shining in the low \gls{snr} regime and outperforming traditional analytical codes in reliability. 

Learned codes for the \gls{su} \gls{awgnf} have been rapidly gaining attention, because, due to their non-linearity, their performance does not suffer as severely with feedback noise as compared to traditional analytical feedback codes. These \gls{su} codes can be classified into two categories: \textit{bit-by-bit} and \textit{symbol-by-symbol} schemes. Bit-by-bit schemes~\cite{kim2018deepcode, safavi2021deep, mashhadi2021drf, shao2023attentioncode} apply deep learning to jointly encode the channel output feedback and message bits, transmitting one bit at a time sequentially (e.g., 50 bits), as introduced in DeepCode~\cite{kim2018deepcode}. In contrast, symbol-by-symbol schemes~\cite{ozfatura2022all, kim2023robust, ankireddy2024lightcode} map a block of message bits to a single symbol and achieve variable code rates by adapting the number of channel uses. 
Symbol-by-symbol schemes reduce the total number of transmissions, including both the forward transmission of codewords and the backward reception of feedback, while achieving improved performance compared to bit-by-bit schemes. Among symbol-by-symbol schemes, {\lightcode}~\cite{ankireddy2024lightcode} achieves the current state-of-the-art performance for the \gls{awgnf}. Compared to other symbol-by-symbol codes, such as the transformer-based Generalized Block Attention Feedback (GBAF)~\cite{ozfatura2022all} and the recurrent neural network (RNN)-based Robust Power-Constrained (RPC) code~\cite{kim2023robust}, {\lightcode} is relatively lightweight, with much fewer parameters and lower complexity. Most of these codes, including the state-of-the-art {\lightcode}, operate at short blocklengths, making them amenable for low-latency applications.

The \glspl{dlecfc} for the \gls{awgnf} outperform traditional linear schemes in terms of reliability at finite blocklengths, particularly when the feedback noise is high and/or the forward \gls{snr} is low. However, only a few deep learning–based implementations have been applied to multi-user channels. For instance, Li \textit{et al.} generalized DeepCode to the fading \gls{awgnbcf} with two users~\cite{li2022deep}. Similarly, Ozfatura \textit{et al.} extended GBAF to the \gls{mac} with feedback~\cite{ozfatura23}.


Deep learning–based codes are \textit{particularly} promising for the \gls{awgnbcf}, where feedback provides capacity gains over the \gls{awgnbc}. In most practical systems, feedback is noisy, and linear \gls{awgnbcf} codes suffer rapid performance degradation even at feedback noise power as low as $-30$ dB. This motivates the use of \glspl{dlecfc} to design robust nonlinear \gls{awgnbcf} codes that maintain high reliability under noisy feedback conditions.

\textit{Main Contributions}: In this paper, we construct finite blocklength codes for the \gls{awgnbcf}, and evaluate them based on their \textit{block error rate} at a fixed code rate. We compare the performance of \textit{both} linear codes and \gls{dlecfc} at various levels of feedback noise. Our contributions are: 
\begin{itemize}
    \item We propose a new linear scheme for the \gls{awgnbcf} which outperforms existing schemes at finite blocklengths with feedback noise.
    \item We propose a general framework for the 2-user \gls{awgnbcf} channel and show how RPC and {\lightcode} are cast into this setting, building on previous works~\cite{malayter2024deep, zhou2024learned}. We demonstrate improved performance over the \gls{tdd} implementation of these codes. We also provide some interpretations of the nonlinear structure, explaining how feedback manifests in error correction.
    \item Finally, we propose a lightweight training scheme to extend the proposed \gls{awgnbcf} codes to an $L>2$ user \gls{awgnbcf}, drawing on insights from the proposed linear scheme. This scheme trains using shared decoder weights, removing the need to train unique decoder weights for each user, saving complexity.
\end{itemize}

The paper is organized as follows. Section~\ref{sec:ProblemSetup} introduces the channel model and key definitions. In Section~\ref{sec:Analytical}, we review existing linear codes and introduce a new linear code for the \gls{awgnbcf}. Section~\ref{sec:Neural} presents the details of the \gls{dlecfc}, together with numerical results and their interpretations. Finally, the paper is concluded in Section~\ref{sec:Conclusion}.

\textit{Notation}: A scalar, vector, matrix, and set is denoted by $x$, $\xvect$, $\mathbf{X}$, and $\scriptX$, respectively. The floor of $x$ is denoted $\lfloor x\rfloor$. We denote $x$ modulo $n$ as $\text{mod}{(x,n)}$. The $n$-th element of $\xvect$ is denoted by $\xvect[n]$, and the $(m,n)$-th element of $\textbf{X}$ is denoted by $\textbf{X}[m,n]$. A matrix $\mathbf{X} = diag(x_1,\cdots,x_n)$ is a square matrix with $(x_1,\cdots,x_n)$ on the diagonal and zeros elsewhere. The $\ell_2$ norm of a vector $\xvect$ is denoted $\|\xvect\|_2$ and the Frobenius norm of a matrix $\mathbf{X}$ is denoted $\|\mathbf{X}\|_F$.
The notation $\xvect\perp\yvect$ indicates that vectors 
 $\xvect$ and $\yvect$ are orthogonal, and $\xvect^T$ is the transpose of $\xvect$. The standard unit vector along the $i$-th axis is denoted by $\evect_i$, and the cardinality of a set $\scriptX$ is denoted by $\card{\scriptX}$. The notation $\{x_n\}_{n=1}^N$ is shorthand for the set $\{x_1, \ldots, x_N\}$. We use $\F_2$ to denote the finite field with elements $\{0, 1\}$, $\R$ for the set of real numbers, and $\R^+$ for the set of positive real numbers. Additionally, $\mathrm{sgn}^*(x) = 1$ if $x \geq 0$, and $-1$ otherwise. The notation $\abs{x}$ denotes the absolute value of $x$. $Q(x)$ is the complementary distribution function of $\mathcal{N}(0,1)$, i.e., $Q(x) = \frac{1}{\sqrt{2\pi}} \int_x^\infty \exp\left(-\frac{u^2}{2}\right) \, du$.

%% file: ProblemSetup.tex

\subsection{Channel Model}
We consider a real-valued $L$-user AWGN broadcast channel with output feedback (\gls{awgnbcf}), as illustrated in Fig.~\ref{fig:SystemModel}. The transmitter jointly encodes $L$ independent, uniformly distributed messages $W_\ell \in \scriptW_\ell$, one intended for each receiver $\ell \in \{1, \ldots,L\}$, into one output message stream, where $\scriptW_\ell$ is the set of all possible messages for user $\ell$. Each message is assumed to satisfy $W_\ell\in \F_2^{K_\ell}$, where $K_\ell$ denotes the message length in bits. A total of $K_1 + \cdots + K_L$ message bits are sent over $N$ channel uses. At the $t$-th channel use, the channel output at receiver $\ell$ is given by 
\begin{equation}
    \yvect_\ell[t] = \xvect[t] + \nvect_\ell^b[t], \quad \ell\in \{1,\ldots, L\}
    \label{eqn:ybroadcast}
\end{equation}
where $\xvect[t] \in \R$ is the transmitted symbol at time $t$, and $\nvect_\ell^b[t]$ is temporally independent and identically distributed (i.i.d) noise with distribution $\nvect^b_\ell[t] \sim \scriptN\paren{0, \sigma_{b}^2}$. The superscript $b$ indicates the noise on the \textit{broadcast} channel. The transmitted symbols satisfy the average power constraint
\begin{equation}
    \frac{1}{N}\expect{\sum_{t=1}^N \xvect^2[t]} \leq P\label{eqn:powerCst}
\end{equation}
where $P$ is the power scaling factor. 

Each receiver sends \textit{passive} feedback to the transmitter, meaning it transmits its most recent received symbol un-encoded through an AWGN channel. The feedback from receiver $\ell$ at time $t$ is given by
\begin{equation}
    \zvect_\ell[t] = \yvect_\ell[t] + \nvect_\ell^f[t]
\end{equation}
where $\nvect^f_\ell[t] \sim \scriptN\paren{0, \sigma_{f}^2}$ and is i.i.d in time.  The superscript $f$  denotes noise on the \textit{feedback} link. 

The signal-to-noise ratio (SNR) on the broadcast channel is defined as
\begin{equation}
    \snrb = \frac{P}{\sigma_b^2}.
\end{equation}
In our results and analysis, we characterize the feedback noise power directly instead of using an SNR metric for the feedback channel.

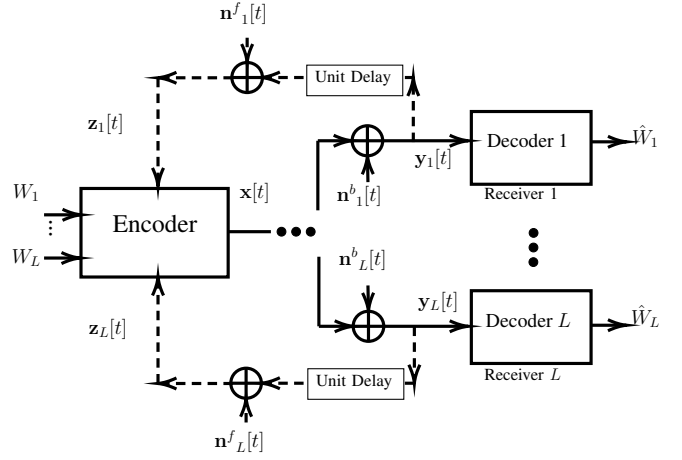
\begin{figure}[!t]
    \centering
        
    \resizebox{1\columnwidth}{!}{%
   \input{Figures/Fig_SystemModel}
}
    \caption{$L$-user AWGN broadcast channel with noisy channel output feedback.}
    \label{fig:SystemModel}
\end{figure}

\subsection{Coding Definitions}
The \textit{rate} of user $\ell$ is $R_\ell\in \R^+$, given by $R_\ell = K_\ell/N$. The \textit{sum-rate} is  $R_{sum} = \sum_{\ell = 1}^L R_{\ell}$. Thus, a $(2^{K_1},  \ldots, 2^{K_L},N)$ \textit{code} for the \gls{awgnbcf} consists of: 
\begin{itemize}
    \item A single encoder represented by a set of $N$ functions $\{f_t\}_{t=1}^N$. For the $t$-th channel use, the encoder $f_t$ maps all users’ messages  $\{W_\ell\}_{\ell=1}^L$ and available feedback $\{\zvect_\ell[1], \ldots, \zvect_\ell[t-1]\}_{\ell = 1}^L$ to a symbol $\xvect[t]\in \R$, ensuring that the power constraint in \eqref{eqn:powerCst} is satisfied. The encoding function at time $t$ is given by
    \begin{equation}
        \xvect[t] = f_t\left(\{W_\ell, \zvect_\ell[1], \ldots, \zvect_\ell[t-1]\}_{\ell=1}^L\right).
    \end{equation}
    \item A set of $L$ decoders, where the $\ell$-th decoder is denoted by a decoding function $g_\ell$.  The $\ell$-th decoder maps the received noisy symbols $\{\yvect_\ell[1], \ldots, \yvect_\ell[N]\}$ to the estimated message $\widehat{W}_\ell\in \scriptW_\ell$, where $\widehat{W}_\ell$ is the decoded output for receiver $\ell$.  After user $\ell$ receives all $N$ symbols, the decoding is defined as 
    \begin{align}
        \widehat{W}_\ell = g_\ell\paren{\yvect_\ell[1], \ldots, \yvect_\ell[N]}.
    \end{align}
\end{itemize}
The average \gls{bler} for receiver $\ell$ is defined as $\mathbb{P}_{e,\ell} \coloneqq \Prob{\widehat{W}_\ell \neq W_\ell}$. 
The design of the encoding functions $\{f_t\}_{t=1}^N$ and the decoding function $g_\ell$ for user $\ell$ can be parameterized using analytical linear codes or neural networks to minimize the overall average \gls{bler} across all users subject to the average power constraint in \eqref{eqn:powerCst}. Specifically, for any given encoder-decoder design, the objective function for constructing an \gls{awgnbcf} code is
\begin{align}
\min_{\{f_t\}_{t=1}^N,\, g_1, \ldots, g_L} \quad 
    & \frac{1}{L} \sum_{\ell=1}^{L} \mathbb{P}_{e,\ell}\label{eqn:optProblem} \\
    \textrm{subject to} \quad 
    & \frac{1}{N}\mathbb{E}_{\{W_\ell, \nvect_\ell^b, \nvect_\ell^f \}_{\ell=1}^L} \left( \sum_{t=1}^N \xvect^2[t] \right) \leq P \nonumber,
\end{align}
where the power constraint expectation is taken over the distributions of the input messages $\{W_\ell\}_{\ell=1}^L$ and the noises $\{\nvect_\ell^b\}_{\ell=1}^L$, $\{\nvect_\ell^f\}_{\ell=1}^L$, since $\xvect[t]$ depends on these quantities. In this work, we focus on the symmetric rate region of the \gls{awgnbcf}, where $R = R_1 = \cdots = R_L$.

%% file: Figures/Fig_SystemModel.tex
\tikzset{every picture/.style={line width=0.75pt}} 

\begin{tikzpicture}[x=0.75pt,y=0.75pt,yscale=-1,xscale=1]

\draw  [line width=2.25]  (91,185) -- (230.5,185) -- (230.5,269) -- (91,269) -- cycle ;
\draw  [line width=2.25]  (461,112) -- (574.5,112) -- (574.5,181) -- (461,181) -- cycle ;
\draw  [line width=2.25]  (461,282) -- (574.5,282) -- (574.5,351) -- (461,351) -- cycle ;
\draw [color={rgb, 255:red, 0; green, 0; blue, 0 }  ,draw opacity=1 ][line width=2.25]    (56,210) -- (87,210) ;
\draw [shift={(91,210)}, rotate = 180] [color={rgb, 255:red, 0; green, 0; blue, 0 }  ,draw opacity=1 ][line width=2.25]    (17.49,-5.26) .. controls (11.12,-2.23) and (5.29,-0.48) .. (0,0) .. controls (5.29,0.48) and (11.12,2.23) .. (17.49,5.26)   ;
\draw [color={rgb, 255:red, 0; green, 0; blue, 0 }  ,draw opacity=1 ][line width=2.25]    (56,251) -- (86,251) ;
\draw [shift={(90,251)}, rotate = 180] [color={rgb, 255:red, 0; green, 0; blue, 0 }  ,draw opacity=1 ][line width=2.25]    (17.49,-5.26) .. controls (11.12,-2.23) and (5.29,-0.48) .. (0,0) .. controls (5.29,0.48) and (11.12,2.23) .. (17.49,5.26)   ;
\draw [color={rgb, 255:red, 0; green, 0; blue, 0 }  ,draw opacity=1 ][line width=2.25]    (231,226) -- (271,226) ;
\draw [color={rgb, 255:red, 0; green, 0; blue, 0 }  ,draw opacity=1 ][line width=2.25]    (315.5,205.75) -- (315.5,139.75) ;
\draw [color={rgb, 255:red, 0; green, 0; blue, 0 }  ,draw opacity=1 ][line width=2.25]  [dash pattern={on 6.75pt off 4.5pt}]  (164,80) -- (164,179) ;
\draw [shift={(164,183)}, rotate = 270] [color={rgb, 255:red, 0; green, 0; blue, 0 }  ,draw opacity=1 ][line width=2.25]    (17.49,-5.26) .. controls (11.12,-2.23) and (5.29,-0.48) .. (0,0) .. controls (5.29,0.48) and (11.12,2.23) .. (17.49,5.26)   ;
\draw [color={rgb, 255:red, 0; green, 0; blue, 0 }  ,draw opacity=1 ][line width=2.25]  [dash pattern={on 6.75pt off 4.5pt}]  (235,80) -- (168,80) ;
\draw [shift={(164,80)}, rotate = 360] [color={rgb, 255:red, 0; green, 0; blue, 0 }  ,draw opacity=1 ][line width=2.25]    (17.49,-5.26) .. controls (11.12,-2.23) and (5.29,-0.48) .. (0,0) .. controls (5.29,0.48) and (11.12,2.23) .. (17.49,5.26)   ;
\draw [color={rgb, 255:red, 0; green, 0; blue, 0 }  ,draw opacity=1 ][line width=2.25]  [dash pattern={on 6.75pt off 4.5pt}]  (164,370) -- (164,273) ;
\draw [shift={(164,269)}, rotate = 90] [color={rgb, 255:red, 0; green, 0; blue, 0 }  ,draw opacity=1 ][line width=2.25]    (17.49,-5.26) .. controls (11.12,-2.23) and (5.29,-0.48) .. (0,0) .. controls (5.29,0.48) and (11.12,2.23) .. (17.49,5.26)   ;
\draw [color={rgb, 255:red, 0; green, 0; blue, 0 }  ,draw opacity=1 ][line width=2.25]  [dash pattern={on 6.75pt off 4.5pt}]  (235,370) -- (168,370) ;
\draw [shift={(164,370)}, rotate = 360] [color={rgb, 255:red, 0; green, 0; blue, 0 }  ,draw opacity=1 ][line width=2.25]    (17.49,-5.26) .. controls (11.12,-2.23) and (5.29,-0.48) .. (0,0) .. controls (5.29,0.48) and (11.12,2.23) .. (17.49,5.26)   ;
\draw [color={rgb, 255:red, 0; green, 0; blue, 0 }  ,draw opacity=1 ][line width=2.25]  [dash pattern={on 6.75pt off 4.5pt}]  (247.5,407) -- (247.5,388) ;
\draw [shift={(247.5,384)}, rotate = 90] [color={rgb, 255:red, 0; green, 0; blue, 0 }  ,draw opacity=1 ][line width=2.25]    (17.49,-5.26) .. controls (11.12,-2.23) and (5.29,-0.48) .. (0,0) .. controls (5.29,0.48) and (11.12,2.23) .. (17.49,5.26)   ;
\draw  [line width=2.25]  (233,369.75) .. controls (233,361.88) and (239.49,355.5) .. (247.5,355.5) .. controls (255.51,355.5) and (262,361.88) .. (262,369.75) .. controls (262,377.62) and (255.51,384) .. (247.5,384) .. controls (239.49,384) and (233,377.62) .. (233,369.75) -- cycle ;
\draw [line width=2.25]    (247.5,355.5) -- (247.5,384) ;

\draw [line width=2.25]    (233,369.75) -- (262,369.75) ;

\draw [color={rgb, 255:red, 0; green, 0; blue, 0 }  ,draw opacity=1 ][line width=2.25]  [dash pattern={on 6.75pt off 4.5pt}]  (407,369.75) -- (266,369.75) ;
\draw [shift={(262,369.75)}, rotate = 360] [color={rgb, 255:red, 0; green, 0; blue, 0 }  ,draw opacity=1 ][line width=2.25]    (17.49,-5.26) .. controls (11.12,-2.23) and (5.29,-0.48) .. (0,0) .. controls (5.29,0.48) and (11.12,2.23) .. (17.49,5.26)   ;
\draw [color={rgb, 255:red, 0; green, 0; blue, 0 }  ,draw opacity=1 ][line width=2.25]  [dash pattern={on 6.75pt off 4.5pt}]  (407,316) -- (407,365.75) ;
\draw [shift={(407,369.75)}, rotate = 270] [color={rgb, 255:red, 0; green, 0; blue, 0 }  ,draw opacity=1 ][line width=2.25]    (17.49,-5.26) .. controls (11.12,-2.23) and (5.29,-0.48) .. (0,0) .. controls (5.29,0.48) and (11.12,2.23) .. (17.49,5.26)   ;
\draw [color={rgb, 255:red, 0; green, 0; blue, 0 }  ,draw opacity=1 ][line width=2.25]    (316.5,315.75) -- (345,315.75) ;
\draw [shift={(349,315.75)}, rotate = 180] [color={rgb, 255:red, 0; green, 0; blue, 0 }  ,draw opacity=1 ][line width=2.25]    (17.49,-5.26) .. controls (11.12,-2.23) and (5.29,-0.48) .. (0,0) .. controls (5.29,0.48) and (11.12,2.23) .. (17.49,5.26)   ;
\draw [color={rgb, 255:red, 0; green, 0; blue, 0 }  ,draw opacity=1 ][line width=2.25]    (315.5,139.75) -- (344,139.75) ;
\draw [shift={(348,139.75)}, rotate = 180] [color={rgb, 255:red, 0; green, 0; blue, 0 }  ,draw opacity=1 ][line width=2.25]    (17.49,-5.26) .. controls (11.12,-2.23) and (5.29,-0.48) .. (0,0) .. controls (5.29,0.48) and (11.12,2.23) .. (17.49,5.26)   ;
\draw  [line width=2.25]  (349,315.75) .. controls (349,307.88) and (355.49,301.5) .. (363.5,301.5) .. controls (371.51,301.5) and (378,307.88) .. (378,315.75) .. controls (378,323.62) and (371.51,330) .. (363.5,330) .. controls (355.49,330) and (349,323.62) .. (349,315.75) -- cycle ;
\draw [line width=2.25]    (363.5,301.5) -- (363.5,330) ;

\draw [line width=2.25]    (349,315.75) -- (378,315.75) ;

\draw  [line width=2.25]  (348,139.75) .. controls (348,131.88) and (354.49,125.5) .. (362.5,125.5) .. controls (370.51,125.5) and (377,131.88) .. (377,139.75) .. controls (377,147.62) and (370.51,154) .. (362.5,154) .. controls (354.49,154) and (348,147.62) .. (348,139.75) -- cycle ;
\draw [line width=2.25]    (362.5,125.5) -- (362.5,154) ;

\draw [line width=2.25]    (348,139.75) -- (377,139.75) ;

\draw [color={rgb, 255:red, 0; green, 0; blue, 0 }  ,draw opacity=1 ][line width=2.25]    (363,277) -- (363.42,297.5) ;
\draw [shift={(363.5,301.5)}, rotate = 268.83] [color={rgb, 255:red, 0; green, 0; blue, 0 }  ,draw opacity=1 ][line width=2.25]    (17.49,-5.26) .. controls (11.12,-2.23) and (5.29,-0.48) .. (0,0) .. controls (5.29,0.48) and (11.12,2.23) .. (17.49,5.26)   ;
\draw [color={rgb, 255:red, 0; green, 0; blue, 0 }  ,draw opacity=1 ][line width=2.25]    (363,179) -- (362.58,158) ;
\draw [shift={(362.5,154)}, rotate = 88.85] [color={rgb, 255:red, 0; green, 0; blue, 0 }  ,draw opacity=1 ][line width=2.25]    (17.49,-5.26) .. controls (11.12,-2.23) and (5.29,-0.48) .. (0,0) .. controls (5.29,0.48) and (11.12,2.23) .. (17.49,5.26)   ;
\draw [color={rgb, 255:red, 0; green, 0; blue, 0 }  ,draw opacity=1 ][line width=2.25]    (377,139.75) -- (454,139.75) ;
\draw [shift={(458,139.75)}, rotate = 180] [color={rgb, 255:red, 0; green, 0; blue, 0 }  ,draw opacity=1 ][line width=2.25]    (17.49,-5.26) .. controls (11.12,-2.23) and (5.29,-0.48) .. (0,0) .. controls (5.29,0.48) and (11.12,2.23) .. (17.49,5.26)   ;
\draw [color={rgb, 255:red, 0; green, 0; blue, 0 }  ,draw opacity=1 ][line width=2.25]    (378,315.75) -- (454,315.75) ;
\draw [shift={(458,315.75)}, rotate = 180] [color={rgb, 255:red, 0; green, 0; blue, 0 }  ,draw opacity=1 ][line width=2.25]    (17.49,-5.26) .. controls (11.12,-2.23) and (5.29,-0.48) .. (0,0) .. controls (5.29,0.48) and (11.12,2.23) .. (17.49,5.26)   ;
\draw  [line width=2.25]  (234,79.75) .. controls (234,71.88) and (240.49,65.5) .. (248.5,65.5) .. controls (256.51,65.5) and (263,71.88) .. (263,79.75) .. controls (263,87.62) and (256.51,94) .. (248.5,94) .. controls (240.49,94) and (234,87.62) .. (234,79.75) -- cycle ;
\draw [line width=2.25]    (248.5,65.5) -- (248.5,94) ;

\draw [line width=2.25]    (234,79.75) -- (263,79.75) ;

\draw [color={rgb, 255:red, 0; green, 0; blue, 0 }  ,draw opacity=1 ][line width=2.25]  [dash pattern={on 6.75pt off 4.5pt}]  (248.5,42) -- (248.5,61.5) ;
\draw [shift={(248.5,65.5)}, rotate = 270] [color={rgb, 255:red, 0; green, 0; blue, 0 }  ,draw opacity=1 ][line width=2.25]    (17.49,-5.26) .. controls (11.12,-2.23) and (5.29,-0.48) .. (0,0) .. controls (5.29,0.48) and (11.12,2.23) .. (17.49,5.26)   ;
\draw [color={rgb, 255:red, 0; green, 0; blue, 0 }  ,draw opacity=1 ][line width=2.25]  [dash pattern={on 6.75pt off 4.5pt}]  (407,79.75) -- (266,79.75) ;
\draw [shift={(262,79.75)}, rotate = 360] [color={rgb, 255:red, 0; green, 0; blue, 0 }  ,draw opacity=1 ][line width=2.25]    (17.49,-5.26) .. controls (11.12,-2.23) and (5.29,-0.48) .. (0,0) .. controls (5.29,0.48) and (11.12,2.23) .. (17.49,5.26)   ;
\draw [color={rgb, 255:red, 0; green, 0; blue, 0 }  ,draw opacity=1 ][line width=2.25]  [dash pattern={on 6.75pt off 4.5pt}]  (406,141) -- (406,83.75) ;
\draw [shift={(406,79.75)}, rotate = 90] [color={rgb, 255:red, 0; green, 0; blue, 0 }  ,draw opacity=1 ][line width=2.25]    (17.49,-5.26) .. controls (11.12,-2.23) and (5.29,-0.48) .. (0,0) .. controls (5.29,0.48) and (11.12,2.23) .. (17.49,5.26)   ;
\draw [color={rgb, 255:red, 0; green, 0; blue, 0 }  ,draw opacity=1 ][line width=2.25]    (574,314.75) -- (604,314.75) ;
\draw [shift={(608,314.75)}, rotate = 180] [color={rgb, 255:red, 0; green, 0; blue, 0 }  ,draw opacity=1 ][line width=2.25]    (17.49,-5.26) .. controls (11.12,-2.23) and (5.29,-0.48) .. (0,0) .. controls (5.29,0.48) and (11.12,2.23) .. (17.49,5.26)   ;
\draw [color={rgb, 255:red, 0; green, 0; blue, 0 }  ,draw opacity=1 ][line width=2.25]    (574,140.75) -- (604,140.75) ;
\draw [shift={(608,140.75)}, rotate = 180] [color={rgb, 255:red, 0; green, 0; blue, 0 }  ,draw opacity=1 ][line width=2.25]    (17.49,-5.26) .. controls (11.12,-2.23) and (5.29,-0.48) .. (0,0) .. controls (5.29,0.48) and (11.12,2.23) .. (17.49,5.26)   ;
\draw  [fill={rgb, 255:red, 0; green, 0; blue, 0 }  ,fill opacity=1 ] (520.02,223.01) .. controls (522.31,223.03) and (524.16,224.9) .. (524.14,227.19) .. controls (524.13,229.49) and (522.26,231.34) .. (519.96,231.32) .. controls (517.67,231.31) and (515.82,229.43) .. (515.84,227.14) .. controls (515.85,224.84) and (517.72,223) .. (520.02,223.01) -- cycle ;
\draw  [fill={rgb, 255:red, 0; green, 0; blue, 0 }  ,fill opacity=1 ] (519.93,236.86) .. controls (522.22,236.87) and (524.07,238.75) .. (524.05,241.04) .. controls (524.04,243.33) and (522.17,245.18) .. (519.87,245.17) .. controls (517.58,245.15) and (515.73,243.28) .. (515.75,240.98) .. controls (515.76,238.69) and (517.63,236.84) .. (519.93,236.86) -- cycle ;
\draw  [fill={rgb, 255:red, 0; green, 0; blue, 0 }  ,fill opacity=1 ] (519.83,250.7) .. controls (522.13,250.72) and (523.98,252.59) .. (523.96,254.89) .. controls (523.95,257.18) and (522.07,259.03) .. (519.78,259.01) .. controls (517.49,259) and (515.64,257.12) .. (515.65,254.83) .. controls (515.67,252.54) and (517.54,250.69) .. (519.83,250.7) -- cycle ;

\draw  [fill={rgb, 255:red, 0; green, 0; blue, 0 }  ,fill opacity=1 ] (276.9,226.01) .. controls (276.9,223.72) and (278.76,221.86) .. (281.05,221.86) .. controls (283.35,221.86) and (285.21,223.72) .. (285.21,226.01) .. controls (285.21,228.31) and (283.35,230.17) .. (281.05,230.17) .. controls (278.76,230.17) and (276.9,228.31) .. (276.9,226.01) -- cycle ;
\draw  [fill={rgb, 255:red, 0; green, 0; blue, 0 }  ,fill opacity=1 ] (290.75,226.01) .. controls (290.75,223.72) and (292.61,221.86) .. (294.9,221.86) .. controls (297.19,221.86) and (299.05,223.72) .. (299.05,226.01) .. controls (299.05,228.31) and (297.19,230.17) .. (294.9,230.17) .. controls (292.61,230.17) and (290.75,228.31) .. (290.75,226.01) -- cycle ;
\draw  [fill={rgb, 255:red, 0; green, 0; blue, 0 }  ,fill opacity=1 ] (304.59,226.01) .. controls (304.59,223.72) and (306.45,221.86) .. (308.75,221.86) .. controls (311.04,221.86) and (312.9,223.72) .. (312.9,226.01) .. controls (312.9,228.31) and (311.04,230.17) .. (308.75,230.17) .. controls (306.45,230.17) and (304.59,228.31) .. (304.59,226.01) -- cycle ;

\draw [color={rgb, 255:red, 0; green, 0; blue, 0 }  ,draw opacity=1 ][line width=2.25]    (316.5,315.75) -- (316.5,249.75) ;

\draw (472,354) node [anchor=north west][inner sep=0.75pt]  [font=\large] [align=left] {Receiver $\displaystyle L$};
\draw (472,183) node [anchor=north west][inner sep=0.75pt]  [font=\large] [align=left] {Receiver 1};
\draw (216,412.4) node [anchor=north west][inner sep=0.75pt]  [font=\Large]  {$\mathbf{n}^{f}{}_{L}^{\ }[ t]$};
\draw (240,178.4) node [anchor=north west][inner sep=0.75pt]  [font=\Large]  {$\mathbf{x}[ t]$};
\draw (406,145.9) node [anchor=north west][inner sep=0.75pt]  [font=\Large]  {$\mathbf{y}_{1}[ t]$};
\draw (25,179) node [anchor=north west][inner sep=0.75pt]  [font=\Large] [align=left] {$\displaystyle W_{1}$};
\draw  [fill={rgb, 255:red, 255; green, 255; blue, 255 }  ,fill opacity=1 ]  (305,68) -- (396,68) -- (396,97) -- (305,97) -- cycle  ;
\draw (308,72) node [anchor=north west][inner sep=0.75pt]  [font=\large] [align=left] {\begin{minipage}[lt]{59.18pt}\setlength\topsep{0pt}
\begin{center}
Unit Delay
\end{center}

\end{minipage}};
\draw  [fill={rgb, 255:red, 255; green, 255; blue, 255 }  ,fill opacity=1 ]  (306,356) -- (397,356) -- (397,385) -- (306,385) -- cycle  ;
\draw (309,360) node [anchor=north west][inner sep=0.75pt]  [font=\large] [align=left] {\begin{minipage}[lt]{59.18pt}\setlength\topsep{0pt}
\begin{center}
Unit Delay
\end{center}

\end{minipage}};
\draw (24,241) node [anchor=north west][inner sep=0.75pt]  [font=\Large] [align=left] {$\displaystyle W_{L}$};
\draw  [draw opacity=0]  (29.02,209.01) -- (69.35,209.01) -- (69.35,239.48) -- (29.02,239.48) -- cycle  ;
\draw (65.03,212.04) node [anchor=north west][inner sep=0.75pt]  [font=\LARGE,rotate=-89.25] [align=left] {...};
\draw (334,239.4) node [anchor=north west][inner sep=0.75pt]  [font=\Large]  {$\mathbf{n}^{b}{}_{L}^{\ }[ t]$};
\draw (332,179.4) node [anchor=north west][inner sep=0.75pt]  [font=\Large]  {$\mathbf{n}^{b}{}_{1}^{\ }[ t]$};
\draw (409,281.9) node [anchor=north west][inner sep=0.75pt]  [font=\Large]  {$\mathbf{y}_{L}[ t]$};
\draw (221,6.4) node [anchor=north west][inner sep=0.75pt]  [font=\Large]  {$\mathbf{n}^{f}{}_{1}^{\ }[ t]$};
\draw (612,123) node [anchor=north west][inner sep=0.75pt]  [font=\Large] [align=left] {$\displaystyle \hat{W}_{1}$};
\draw (611,294) node [anchor=north west][inner sep=0.75pt]  [font=\Large] [align=left] {$\displaystyle \hat{W}_{L}$};
\draw (96,112.9) node [anchor=north west][inner sep=0.75pt]  [font=\Large]  {$\mathbf{z}_{1}[ t]$};
\draw (97,306.9) node [anchor=north west][inner sep=0.75pt]  [font=\Large]  {$\mathbf{z}_{L}[ t]$};
\draw (114,208) node [anchor=north west][inner sep=0.75pt]  [font=\LARGE] [align=left] {\begin{minipage}[lt]{68.45pt}\setlength\topsep{0pt}
\begin{center}
Encoder
\end{center}

\end{minipage}};
\draw (466,131) node [anchor=north west][inner sep=0.75pt]  [font=\Large] [align=left] {\begin{minipage}[lt]{70.5pt}\setlength\topsep{0pt}
\begin{center}
Decoder 1
\end{center}

\end{minipage}};
\draw (466,302) node [anchor=north west][inner sep=0.75pt]  [font=\Large] [align=left] {\begin{minipage}[lt]{71.57pt}\setlength\topsep{0pt}
\begin{center}
Decoder $\displaystyle L$
\end{center}

\end{minipage}};

\end{tikzpicture}

%% file: Analytical.tex

In this section, we review existing analytical coding schemes for the \gls{awgnbcf} and introduce a new linear scheme that improves robustness to feedback noise at non-asymptotic blocklengths. Specifically, we outline three classes of linear codes for the \gls{awgnbcf}: the \gls{ol} scheme, the \gls{lqg} scheme, and \gls{snr}-maximizing schemes.

\subsection{Ozarow-Leung Scheme}
We begin by describing the \gls{ol} scheme, which adapts the \gls{lmmse}-based Schalkwijk-Kailath (SK) scheme~\cite{schalkwijk1966coding} from the AWGN-F to the \gls{awgnbcf} setting for 2 users~\cite{ozarowAchievable}. The \gls{ol} scheme was the first to demonstrate the enlarged capacity region enabled by feedback in the \gls{awgnbcf}, achieving doubly exponential error decay with blocklength $N$~\cite{ozarowAchievable}. The \gls{ol} scheme was extended by Kramer to support more than two receivers, and we refer readers to~\cite{Kramer2002} for details. Next, we introduce the \gls{eol} scheme, which incorporates memory into the estimator used in the \gls{ol} scheme~\cite{murin2014ozarow}, resulting in both improved achievable rates and enhanced reliability.

\subsubsection{\gls{ol} scheme} 
In the \gls{ol} scheme with two receivers ($L = 2$), the messages $W_1 \in \scriptW_1$ and $W_2 \in \scriptW_2$ are mapped to real numbers and transmitted separately over the broadcast channel. Each receiver estimates its message based on the received noisy symbols. In subsequent rounds, the transmitter computes and sends a linear combination of the current user errors using the noiseless feedback. The receivers then update their estimates, gradually reducing the errors.

Specifically, the message bits $W_\ell \in \scriptW_\ell$, where $\ell \in \{1, 2\}$, are mapped to a \gls{pam} symbol $\Theta_\ell$ from the constellation $\{\pm{1\eta}, \pm{3\eta}, \dots, \pm{(2^{K_\ell}-1)\eta}\}$, with $\eta = \sqrt{\frac{3}{2^{2K_\ell}-1}}$ to satisfy the unit power constraint. The estimated message at receiver $\ell$ at time $t$ is denoted by $\widehat{\Thetavect}_{\ell}[t]$, with the estimation error $\epsvect_{\ell}[t] = \widehat{\Thetavect}_{\ell}[t] - \Theta_\ell$ and corresponding mean squared error  $\avect_{\ell}[t] = \expect{\epsvect_{\ell}^2[t]}$. The correlation coefficient between the estimation errors at the two receivers is defined as $\rhovect[t] = \frac{\mathbb{E}(\epsvect_{1}[t] \epsvect_{2}[t])}{\sqrt{\avect_{1}[t] \avect_{2}[t]}}$. In the first two rounds, the transmitter sends the \gls{pam} symbols separately, each with power $P$, given as
\begin{equation}
    \xvect[t] = \sqrt{P}\,\Theta_t, \quad t \in \{1, 2\}.
\end{equation}

After the first two transmissions, receiver 1 disregards the second transmission, while receiver 2 ignores the first. 
The receivers estimate their \gls{pam} symbols using
\begin{equation}
    \widehat{\Thetavect}_\ell[2] = \frac{\sqrt{P}}{P + \sigma_b^2} \, \yvect_\ell[\ell], \quad \ell \in \{1, 2\}.
\end{equation}
For the $t$-th transmission, where $t \geq 3$, the transmitter sends
\ifdraft{
\begin{align}
    \xvect[t] &= \sqrt{\frac{P}{D_{t-1}}}\Bigg[\frac{\epsvect_{1}[t-1]}{\sqrt{\avect_{1}[t-1]}} \label{eqn:OLsymbol} + \frac{\epsvect_{2}[t-1]}{\sqrt{\avect_{2}[t-1]}}g\,\mathrm{sgn}^*(\rhovect[t-1])\Bigg], 
\end{align}
}{
\begin{align}
    \xvect[t] &= \sqrt{\frac{P}{D_{t-1}}}\Bigg[\frac{\epsvect_{1}[t-1]}{\sqrt{\avect_{1}[t-1]}} \label{eqn:OLsymbol} \\&+ \frac{\epsvect_{2}[t-1]}{\sqrt{\avect_{2}[t-1]}}g\,\mathrm{sgn}^*(\rhovect[t-1])\Bigg], \nonumber
\end{align}}
where $D_{t-1} = 1 + g^2 + 2g\abs{\rhovect[t-1]}$, and $g \ge 0$ balances the trade-off between the two users\footnote{We set $g=1$ to achieve similar BLER for both receivers.}.


At receiver $\ell \in \{1,2\}$, a \gls{mmse} estimator uses $\yvect_\ell[t]$ to estimate $\epsvect_{\ell}[t-1]$, and updates the symbol as
\begin{equation}
\widehat{\Thetavect}_{\ell}[t] =   \widehat{\Thetavect}_{\ell}[t-1] - \frac{\expect{\epsvect_{\ell}[t-1]\yvect_\ell[t]}}{\expect{\yvect_\ell^2[t]}}\yvect_\ell[t].
\end{equation}
The required expectations are given by
\ifdraft{\begin{align*}
    &\expect{\yvect_\ell^2[t]} = P + \sigma_{b}^2, \\
    &\expect{\epsvect_{1}[t-1]\, \yvect_1[t]} = \sqrt{\frac{P}{D_{t-1}}} \sqrt{\avect_{1}[t-1]}\,(1 + g\abs{\rhovect[t-1]}), \\
    &\expect{\epsvect_{2}[t-1]\, \yvect_{2}[t]} = \sqrt{\frac{P}{D_{t-1}}} \sqrt{\avect_{2}[t-1]}\,(g + \abs{\rhovect[t-1]})\,\mathrm{sgn}^*(\rhovect[t-1]). 
\end{align*}}{
\begin{equation*}
\resizebox{\columnwidth}{!}{$
\begin{aligned}
    &\expect{\yvect_\ell^2[t]} = P + \sigma_{b}^2, \\
    &\expect{\epsvect_{1}[t-1]\, \yvect_1[t]} = \sqrt{\frac{P}{D_{t-1}}} \sqrt{\avect_{1}[t-1]}\,(1 + g\abs{\rhovect[t-1]}), \\
    &\expect{\epsvect_{2}[t-1]\, \yvect_{2}[t]} = \sqrt{\frac{P}{D_{t-1}}} \sqrt{\avect_{2}[t-1]}\,(g + \abs{\rhovect[t-1]})\,\mathrm{sgn}^*(\rhovect[t-1]). 
\end{aligned}
$}
\end{equation*}}



At a fixed code rate, the performance of the \gls{ol} scheme depends on the message length $K$, which is also observed in other analytical and learned codes. At low SNR, shorter messages perform better due to larger constellation spacing, while at high SNR, longer messages improve performance since the decoding error probability decreases doubly exponentially in blocklength. {However, $K$ cannot be too large due to precision issues and quantization errors associated with $2^{K}$ \gls{pam} modulation. In our experiments, when $K \geq 24$, the BLER starts to increase rather than decrease. To accommodate longer lengths $T$, we select an appropriate value for $K$ to mitigate precision issues by using a small $K$ at low SNR and a large $K$ at high SNR. The length $T$ is then divided into $m = T/K$ chunks of bits, each with message bit length $K$, and each chunk is encoded using the chosen coding scheme. The deep-learned codes also adhere to this rule. However, as the number of rounds increases, the input space expands, which raises the difficulty of training.}  


\subsubsection{\gls{eol} scheme}
 The \gls{eol} scheme~\cite{murin2014ozarow} extends the \gls{ol} scheme by incorporating an \gls{mmse} estimator that exploits both the current and previous outputs, $ \mathbf{Y}_\ell[t] \coloneqq [\yvect_\ell[t], \yvect_\ell[t-1]]^{T} $. Let $\mathbf{Q}_{\ell}[t]$ be the covariance matrix of $\mathbf{Y}_\ell[t]$, then the MMSE estimator is  $\hat{\epsvect}_{\ell}[t-1] = \expect{\epsvect_{\ell}[t-1] | \mathbf{Y}_\ell[t]} = {\expect{\epsvect_{\ell}[t-1]\cdot\mathbf{Y}_\ell^{T}[t]}}\cdot\mathbf{Q}_{\ell}^{-1}[t]\cdot\mathbf{Y}_\ell[t]$. Explicit forms are given in~\cite{murin2014ozarow}.
  


\subsection{Linear Control Based Schemes}

In addition to the \gls{ol} scheme, control-theoretic codes have been developed for the \gls{awgnbcf}. The \gls{lqg} code~\cite{ardestanizadeh2012lqg} is a linear code that achieves a larger rate region than the \gls{ol} scheme. It iteratively refines the receivers' estimates by transmitting a linear combination of estimation errors, choosing estimates to minimize the steady-state power of the channel input rather than using the \gls{mmse}.

For the \gls{lqg} code\footnote{Here, we consider an $L$-receiver real \gls{awgnbcf}.}, each message $W_\ell \in \scriptW_\ell$ is mapped to a \gls{pam} symbol $\Theta_\ell$, uniformly distributed over $[0, 1]$. The vector of symbols is $\bm{\Theta} = [\Theta_1, \ldots, \Theta_L]^{T}$. The linear dynamical system can be described by a matrix $\mathbf{S}\in \R^{L\times N}$, which stores the state of each receiver over time. Letting $\svect_t$ denote the $t$-th column of $\mathbf{S}$, the system evolves according to
\begin{align}
    \mathbf{s}_1 &= \bm{\Theta} \nonumber\\
    \mathbf{s}_t &= \mathbf{A}\mathbf{s}_{t-1} + \mathbf{r}_{t-1}, \quad t = \{2,\ldots, N\},
\end{align}
where $\mathbf{s}_t = [\mathbf{S}[1,t], \ldots, \mathbf{S}[L,t]]^T \in \mathbb{R}^L$ represents the system state at time $t \in \{1,\ldots, N\}$, and $\mathbf{A} = \text{diag}(a_1, \ldots, a_L)\in \mathbb{R}^{L\times L}$ with $\abs{a_i} > 1$. The vector $\mathbf{r}_t$ stores the channel output for each receiver at time $t$, given by $\mathbf{r}_t = \mathbf{b}\xvect[t] + \mathbf{z}_t$, where $\mathbf{b} = [1, \ldots, 1]^{T} \in \mathbb{R}^{L}$ is the channel gain and the $\ell$th index of $\mathbf{z}_t\in \R^L$ is the $\ell$th receiver's noise realization at time $t$.

Given the system, the encoder transmits the symbol
\begin{equation}
    \xvect[t]  = -\mathbf{c}\mathbf{s}_t
\end{equation}
where $\mathbf{c}= (\mathbf{b}^{T}\mathbf{G}\mathbf{b} + 1)^{-1}\mathbf{b}^T\mathbf{GA} \in \mathbb{R}^{1\times L}$, and $\mathbf{G}$ is the unique positive-definite solution to the discrete algebraic Riccati equation. Consequently, the receiver decoder $\ell$ computes the estimates as  
\begin{equation}
    \widehat{\Theta}_\ell = -a_\ell^{-t}\widehat{\mathbf{S}}[\ell,t+1],
\end{equation}
where $\widehat{\mathbf{S}}[\ell,t]$ is updated recursively as
\begin{align}
    \widehat{\mathbf{S}}[\ell,1] &= 0, \nonumber\\
    \widehat{\mathbf{S}}[\ell,t] &= a_\ell\widehat{\mathbf{S}}[\ell,t-1] + \mathbf{z}_{t-1}[\ell].
\end{align}
In the case of two users with uncorrelated noise, $\textbf{A} = diag\paren{a,-a}$ where $a$ is chosen to satisfy the asymptotic average power constraint 
\begin{align*}
    \frac{(a^4-1)(a^2+1)}{2a^2} = \frac{P}{\sigma_b^2},
\end{align*}
but in the short blocklength regime, $a$ is found computationally to satisfy the power constraint.

For the real $L$-receiver \gls{awgnbcf} with independent noise, the \gls{lqg} scheme achieves a symmetric rate $R_1 = \cdots = R_L$ under power constraint $P$ given by \cite{ardestanizadeh2012lqg}
\begin{align}
    R<\frac{1}{2L}\log_2(1 + P \phi), \label{eqn:LQGbound}
\end{align}
where $\phi\in [1,L]$ is the unique solution to 
\begin{align*}
    (1+P\phi)^{L-1} = \paren{1 + \frac{P}{L}\phi(L-\phi)}^L.
\end{align*}
Here, $\phi$ quantifies the cooperation among receivers enabled by feedback~\cite{ardestanizadeh2012lqg}.

\subsection{\gls{snr}-Maximizing Linear Schemes} 
 In addition to the \gls{lqg} and \gls{mmse}-based schemes, another line of work aims to maximize the effective \gls{snr} at each receiver.  In the single user case, the \gls{cl} scheme~\cite{chance2011concatenated} uses a concatenated structure with a linear inner code optimized for the received \gls{snr}. Similar constructions can be extended to the broadcast setting, where Ahmed \textit{et al.}~\cite{ahmad2015concatenated} propose a linear code, which we refer to as the ACLW scheme. However, the ACLW scheme performs poorly at short blocklengths. To address this limitation, we propose the \gls{bcl} scheme, which improves reliability performance at short blocklengths while achieving the \gls{lqg} rate in~\eqref{eqn:LQGbound}.

\subsubsection{SNR-Maximizing Encoding Matrix Form}\label{sec:ourscheme}
Like the \gls{ol} scheme, the intended message for receiver $\ell$, $W_\ell$, is mapped to a \gls{pam} constellation symbol $\Theta_\ell$ which is uniformly distributed. The constellation is power-constrained so that
\begin{align}
\expect{\abs{\Theta_\ell}^2} = (1-\gamma)\frac{NP}{L}.
\end{align}
We refer to $\gamma\in(0,1)$ as a \textit{power sharing parameter}, which balances the power spent on transmitting the information symbol versus the power spent on noise cancellation. The first $L$ transmit symbols are given by 
\begin{align}
    \xvect[t] = \Theta_t, \ t \in \{1, \ldots, L\}.
\end{align}
The feedback information is encoded for the remaining channel uses.

Let $\hat{N} = N-L$, which denotes the number of channel uses for noise cancellation. The received signal for user $\ell$ can then be expressed in vector form by omitting the first $L$ channel uses, except for the $\ell$-th use, which carries the symbol intended for user $\ell$, as 
\ifdraft{\begin{align}
    \yvect_\ell &= \evect_1 \Theta_\ell + \paren{\eye + \Fmat_\ell}\nvect_\ell^b + \Fmat_\ell \nvect_\ell^f \label{eqn:rxSig_analytical} + \sum_{\ell'\in \{1,\cdots, L\}, \ell'\neq \ell} \Fmat_{\ell'}\paren{\nvect_{\ell'}^b + \nvect_{\ell'}^f}, 
\end{align}}{
\begin{align}
    \yvect_\ell &= \evect_1 \Theta_\ell + \paren{\eye + \Fmat_\ell}\nvect_\ell^b + \Fmat_\ell \nvect_\ell^f \label{eqn:rxSig_analytical} \\
    &+ \sum_{\ell'\in \{1,\cdots, L\}, \ell'\neq \ell} \Fmat_{\ell'}\paren{\nvect_{\ell'}^b + \nvect_{\ell'}^f}, \nonumber
\end{align}}
where $\yvect_\ell \in \R^{(\hat{N}+1)\times 1}$, $\evect_1 \in \R^{(\hat{N}+1)}\times 1$, and $\Fmat_\ell
\in \R^{(\hat{N}+1)\times(\hat{N}+1)}$ is the encoding matrix for user $\ell$. To maintain causality, $\Fmat_\ell$ is strictly lower triangular with zeros on the main diagonal.

After $N$ channel uses, receiver $\ell$ uses a linear combiner $\qvect_\ell \in \R^{(\hat{N}+1)\times 1}$ to estimate the symbol, given by
\begin{align}
    \hat{\Theta}_\ell = \qvect_\ell^T \yvect_\ell. \label{eqn:thetaHat}
\end{align}
Each receiver observes the same input–output relation, consisting of the transmitted signal corrupted by correlated Gaussian noise. Therefore, $\Fmat_\ell$ and $\qvect_\ell$ should be designed to maximize the \gls{snr} at the output of the combiner for each receiver. 

The ACLW scheme assumes each user has the same rate $R$, also called the symmetric rate region.  The inner code of the concatenated scheme is constructed using a general encoding matrix of the form
\begin{align}
    \Fmat_\ell = \Cmat_\ell \Fmat, \label{eqn:genEncMTX}
\end{align}
where $\Fmat$ is a base encoding matrix that is the same for all users due to the rate symmetry. Each $\Cmat_\ell$ is constructed as
\begin{align}
    \Cmat_\ell[i,j] = \begin{cases}
                    {\cvect}_\ell[\text{mod}(i,L)], & i=j\\
                    0, & i\neq j,
                    \end{cases} \label{eqn:Cdesign}
\end{align}
where ${\cvect}_\ell$ is the $\ell$-th row of a $L\times L$ Hadamard matrix. The design of $\Cmat_\ell$ is to null interference between users. In the real channel, this spreading code design limits $L$ to a power of 2, but in complex channels, complex Hadamard matrices can be used, so that any $L\geq2$ works~\cite{ahmad2015concatenated}. 

\begin{figure*}
\begin{align}
SNR_\ell(\gamma) = \frac{\frac{1}{L}PN(1-\gamma)}{\sigma_b^2\|\qvect^T\paren{\eye +\Fmat}\|_2^2 + (\sigma_b^2+\sigma_f^2)\sum_{{i = 1,i\neq \ell}}^L\paren{\|\qvect_\ell^T\Cmat_{i}\Fmat\|_2^2} + \sigma_f^2\|\qvect^T\Fmat\|_2^2}, \label{eqn:SNRourScheme}
\end{align}
\end{figure*}

With $\Fmat$ in the form of~\eqref{eqn:genEncMTX}, the \gls{snr} at receiver $\ell$ is given by~\eqref{eqn:SNRourScheme}.  Maximizing this \gls{snr} jointly over parameters $\qvect$, $\Fmat$, and $\gamma$ is an intractable problem. The ACLW scheme uses a step-by-step method that, for a fixed combiner $\qvect$, constructs an encoding matrix $\Fmat$ with structural assumptions designed to null inter-user interference and maximize \gls{snr} using Lagrange multipliers. From our experiments, the inner code performs poorly at short blocklengths and requires optimization over numerous parameters. Although concatenation can be used on this inner code to achieve rate tuples beyond the perfect feedback capacity region, this requires blocklengths that exceed the target blocklengths of this paper.

\subsubsection{Broadcast Malayter-Chance-Love (BMCL)} \label{sec:proposedCode}
We now introduce the Broadcast Malayter-Chance-Love (BMCL) scheme. We consider the symmetric rate region and assume the received signal is of the form in \eqref{eqn:rxSig_analytical}. The symbol estimate is produced by using a linear combiner as in \eqref{eqn:thetaHat}. We first note that the \gls{ol} scheme and ACLW scheme share a spreading-code like structure. For the \gls{ol} scheme, this appears in the $\mathrm{sgn}^*(\cdot)$ term in \eqref{eqn:OLsymbol}, while the ACLW scheme explicitly uses spreading codes in \eqref{eqn:genEncMTX}.  Using the notion that these schemes utilize a spreading code-like structure and generally exhibit doubly exponential error decay (as in the schemes in \cite{ozarowAchievable,schalkwijk1966coding}, for example), we design a lower-triangular encoding matrix parameterized by a variable $\beta$ and the number of users $L$. This $\beta$-parameterized structure resembles the structure of the \gls{cl} scheme, which allows convenient matrix power scaling. Our matrix and combiner design produce the output \gls{snr} expression in \eqref{eqn:SNRourScheme}. We design $\Fmat$ and $\qvect$ so that the first two terms in the denominator of \eqref{eqn:SNRourScheme} decay to 0 with increasing blocklength, and the third term stays somewhat small. Nonetheless, the third term in \eqref{eqn:SNRourScheme} will still grow unbounded with blocklength in the presence of feedback noise due to noise accumulation, as this is inevitable in any linear \gls{awgnbcf} code. 

 We assume a general form of $\Fmat_\ell$ given by
\begin{align}
    \Fmat_\ell = \Cmat_\ell \Fmat \Cmat_\ell^T, \label{eqn:Fspread}
\end{align} 
where $\Cmat_\ell\in\R^{(\hat{N}+1)\times (\hat{N}+1)}$ is constructed by \eqref{eqn:Cdesign}, and $\Fmat_\ell$ is strictly lower triangular.  We also assume a general form of $\qvect_\ell$ as
\begin{align}
    \qvect_\ell^T = \qvect^T\Cmat_\ell,  \label{eqn:qvectwithspreading}
\end{align}
where $\qvect$ is a base combining vector that is the same for all users due to rate symmetry. Since $\Fmat_\ell$ has the same Frobenius norm as $\Fmat$, the base matrix $\Fmat$ determines the transmit power of each user. Accordingly, $\Fmat$ is constrained such that 
\begin{align}
    \| \Fmat \|_F^2 \leq \frac{NP\gamma}{L(\sigma_b^2+\sigma_f^2)}. \label{eqn:FpowConst}
\end{align}


The base encoding matrix $\Fmat\in\R^{(\hat{N}+1)\times (\hat{N}+1)}$ is structured as 
\begin{align}
    \Fmat[t,m] = \begin{cases} 
                0, & t \leq m\\
                \frac{-\paren{1-\beta^{2 L}}}{L\beta}\beta^{L\lfloor\frac{t - m - 1}{L}\rfloor -\text{mod}(t-m-1,L)}, & t > m
                \end{cases}, \label{eqn:Fbase}
\end{align}
where $\beta\in (0,1)$. In the finite blocklength case and for a fixed $\gamma$, we design $\beta$ such that the matrix in \eqref{eqn:Fbase} satisfies \eqref{eqn:FpowConst} and search for $\beta$ via bisection. The power of $\Fmat$ is given by 
\ifdraft{\begin{align}
    \norm{\Fmat}^2_F = \left(\frac{\beta^{2L}-1}{L\beta} \right)^2\sum_{k=0}^{(\hat{N}-1)}\paren{\hat{N}-k} \beta^{2k-4\text{mod}(k,L)}\label{eqn:Fpower}
\end{align}}{
\begin{equation}
\resizebox{\columnwidth}{!}{$
\begin{aligned}
    \norm{\Fmat}^2_F = \left(\frac{\beta^{2L}-1}{L\beta} \right)^2\sum_{k=0}^{(\hat{N}-1)}\paren{\hat{N}-k} \beta^{2k-4\text{mod}(k,L)}\label{eqn:Fpower}
\end{aligned}
$}
\end{equation}}
For asymptotic blocklengths, \eqref{eqn:Fpower} becomes 
\begin{align}
    \underset{N\to\infty}{\lim} \frac{\norm{\Fmat}^2_F}{N} = \frac{(1-\beta^{2L})^2}{L^2\beta^{2L}(1-\beta^2)}, \label{eqn:matrixPowerLimit}
\end{align}
and, for asymptotic blocklengths, we design $\beta$ to satisfy the equation
\begin{align}
     \frac{(1-\beta^{2L})^2}{L^2\beta^{2L}(1-\beta^2)} = \frac{P\gamma}{L(\sigma_b^2+\sigma_f^2)}\label{eqn:betapower}
\end{align}
which guarantees that the power constraint is satisfied asymptotically for a fixed $\gamma$. The following lemma is important for deriving the BMCL scheme maximum achievable rate. 
\begin{lemma}The asymptotic power of the encoding matrix $\Fmat$ \eqref{eqn:matrixPowerLimit} is strictly decreasing with $\beta\in(0,1)$.
\end{lemma}
\begin{proof}
    See appendix. 
\end{proof} 

 With $\qvect_\ell$ in the form of \eqref{eqn:qvectwithspreading}, the noise covariance for any user is given as 
\begin{align}
    \Rmat = \sigma_b^2\paren{\eye + \Fmat+\Fmat^{T}} + \paren{\sigma_b^2 + \sigma_f^2}\paren{\sum_{\ell'=1}^L\Fmat_{\ell'}\Fmat_{\ell'}^{T}}. \label{eqn:NoiseCov}
\end{align}
Thus, the \gls{snr}-maximizing $\qvect$ is given as 
\begin{align}
  \qvect = \frac{\Rmat^{-1}\evect_1}{\evect_1^T\Rmat^{-1}\evect_1}. \label{eqn:qopt}
\end{align}
Asymptotically, the combiner $\qvect$ in \eqref{eqn:qopt} is
\begin{align}
    \qvect_{\infty}= [1, \beta, \beta^2, \ldots, \beta^{\hat{N}}]^T. \label{eqn:qasym}
\end{align}

We now analyze the \gls{snr} using the $\Fmat$ in \eqref{eqn:Fbase} and form of $\qvect$ in \eqref{eqn:qasym}.  Observing the denominator in \eqref{eqn:SNRourScheme}, we let $\boldsymbol{\psi} = \qvect^T\paren{\eye +\Fmat}$ and let $\boldsymbol{\zeta}_{\ell,i} = \qvect_\ell^T\Cmat_i\Fmat$. Then $\|\boldsymbol{\psi}\|_2^2$ is given by
\ifdraft{\begin{align}
&\|\boldsymbol{\psi}\|_2^2\label{eqn:psiTerm}
= \beta^{2\hat N} + 
\sum_{d=0}^{\hat N-1}
\beta^{2(\hat N - d - 1)+4L\,\big\lfloor d/L \big\rfloor}
\left( 1 - \frac{\text{mod}\paren{d,L}+1}{L}\big(1-\beta^{2L}\big) \right)^{2}
\end{align}}{
\begin{align}
&\|\boldsymbol{\psi}\|_2^2\label{eqn:psiTerm}
= \beta^{2\hat N} + \\\nonumber
&
\sum_{d=0}^{\hat N-1}
\beta^{2(\hat N - d - 1)+4L\,\big\lfloor d/L \big\rfloor}
\left( 1 - \frac{\text{mod}\paren{d,L}+1}{L}\big(1-\beta^{2L}\big) \right)^{2}
\end{align}}
and $\sum_{i=1,i\neq \ell}^L\|\boldsymbol{\zeta}_{\ell,i}\|_2^2$ is 
\ifdraft{\begin{align}
    &\sum_{i=1,i\neq \ell}^L\|\boldsymbol{\zeta}_{\ell,i}\|_2^2 =\label{eqn:zetaTerm} \frac{(1-\beta^{2L})}{L^{2}}
\sum_{r=0}^{L-1}
\bigl(Lr - r^{2}\bigr)\,
\beta^{2(\hat N - r)}\,
\paren{1 - \beta^{2L\,\paren{\left\lfloor \tfrac{\hat N - r}{L}\right\rfloor +1}}}. 
\end{align}}{
\begin{align}
    &\sum_{i=1,i\neq \ell}^L\|\boldsymbol{\zeta}_{\ell,i}\|_2^2 =\label{eqn:zetaTerm} \\&\frac{(1-\beta^{2L})}{L^{2}}
\sum_{r=0}^{L-1}
\bigl(Lr - r^{2}\bigr)\,
\beta^{2(\hat N - r)}\,
\paren{1 - \beta^{2L\,\paren{\left\lfloor \tfrac{\hat N - r}{L}\right\rfloor +1}}}. \nonumber
\end{align} }

In the perfect feedback case, the total noise power for any user is given by 
\begin{align*}
   \sigma_{tot}^2 = \sigma_b^2 \paren{\|\boldsymbol{\psi}\|_2^2 + \sum_{i=1,i\neq \ell}^L\|\boldsymbol{\zeta}_{\ell,i}\|_2^2},
\end{align*}
which, for large $\hat{N}$, is approximately 
\begin{align}
    \sigma^2_{tot} \approx c_1\beta^{2\hat{N}}, \label{eqn:NPlargeN}
\end{align}
where $c_1$ is a positive constant.
Finally, the average \gls{bler} for user $\ell$ is
\begin{align}
    \mathbb{P}_{e,\ell} =2\paren{1-\frac{1}{2^{2NR_\ell}}} Q\paren{\sqrt{\frac{6}{(2^{2NR_{\ell}}-1)} SNR_\ell(\gamma)}}. \label{eqn:newbler}
\end{align}
 Now, we can derive the maximum achievable rate of the proposed \gls{bcl} code. 


\begin{lemma}\label{lemma:capacity}
    The $L$-user \gls{bcl} maximum sum rate for an SNR of $\frac{P}{\sigma_b^2}$ and perfect feedback is given by $$C_{sum}\paren{\frac{P}{\sigma_b^2}} = -L\log
    _2\paren{\beta_{\infty}},$$ where $$\beta_\infty = \paren{\beta : \frac{(1-\beta^{2L})^2}{L^2\beta^{2L}(1-\beta^2)} = \frac{P}{\sigma_b^2 L}}.$$ 
\end{lemma}

\noindent\textit{Proof.} See appendix.

Though the maximum achievable rate of the \gls{bcl} scheme appears to benefit from a linear capacity scaling as the number of users increases, the capacity has a finite bound as the number of users tends to infinity. This limitation is due to the power constraint captured in $\beta_\infty$. The following remark gives the \gls{bcl} maximum sum rate as $L\to\infty$. 

\begin{remark}\label{remark:infiniteusers}
    The \gls{bcl} maximum sum rate for an SNR of $\frac{P}{\sigma_b^2}$ and perfect feedback has a finite limit as the number of users $L$ goes to infinity. That is, 
\begin{align*}
    \underset{L\to\infty}{\lim}C_{sum}\paren{\frac{P}{\sigma_b^2}} = \frac{\alpha}{\ln2} 
\end{align*}
where $\alpha$ satisfies
\begin{align*}
    \frac{(1-e^{-2\alpha})^2}{2\alpha e^{-2\alpha}} = \frac{P}{\sigma_b^2}.
\end{align*}
\end{remark}

\noindent\textit{Proof.} See appendix.

Figure \ref{fig:infiniteUserCapacity} shows the behavior of $C_{sum}\paren{\frac{P}{\sigma_b^2}}$ as a function of the number of users $L$ in the perfect feedback case for various \glspl{snr}. Initially, it can be seen that as the number of users increases, a sum-rate gain is observed, but this eventually flattens to the limit derived in Remark \ref{remark:infiniteusers}. Nonetheless, in all of the cases, there is a sum-rate gain relative to the capacity of the single-user AWGN, with the largest proportional gain corresponding to higher SNR (10 dB). 
\begin{figure}
    \centering
    \includegraphics[width=0.9\linewidth]{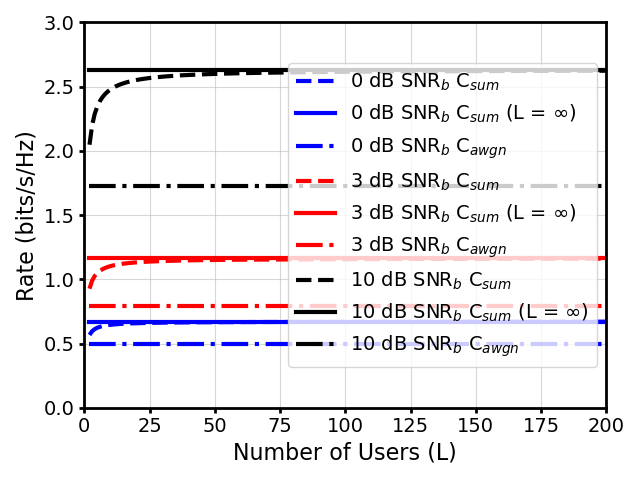}
    \caption{$C_{sum}\paren{\frac{P}{\sigma_b^2}}$ as a function of the number of users $L$ with perfect feedback. Also shown is $\underset{L\to\infty}{\lim}C_{sum}\paren{\frac{P}{\sigma_b^2}}$ from Remark \ref{remark:infiniteusers} and the capacity of the single-user AWGN channel, denoted $C_{awgn}$.}
    \label{fig:infiniteUserCapacity}
\end{figure}

 It can also be shown that $C_{sum}\paren{\frac{P}{\sigma_b^2}}$ is equal to the \gls{lqg} bound in \eqref{eqn:LQGbound} (see Lemma 4 in \cite{ahmad2015concatenated}). This implies that in some regions, the BMCL scheme is optimal over the \gls{ol} scheme in terms of rate \cite{ardestanizadeh2012lqg}. Furthermore, this implies that our scheme, like the \gls{lqg} scheme, is also sum-rate optimal among all linear feedback codes for the symmetric single-antenna \gls{awgnbcf} with equal channel gains \cite{amorDuality}.


In the finite blocklength case, the \gls{snr} expression in \eqref{eqn:SNRourScheme} can be optimized over $\gamma\in(0,1)$ to improve \gls{bler} performance by fixing $\gamma$ and solving for the corresponding $\beta$ that satisfies the power constraint in \eqref{eqn:FpowConst}. The procedure is straightforward in practice and essential for good \gls{bler} performance. Figure~\ref{fig:OurSchemeSumR23noiseless} shows the reliability performance with noiseless feedback for $L=2$ users and $K_1=K_2=K$, with blocklength $N=3K$. The proposed \gls{bcl} code benefits from a blocklength gain, approaching the capacity limit in Lemma~\ref{lemma:capacity}. We also evaluate its performance at short blocklengths with noisy feedback. As shown in Fig.~\ref{fig:OurSchemeSumR23noisy}, the performance degrades with increasing feedback noise but remains reasonably good for low noise ($\sigma_f^2=-30$ dB) and higher forward \gls{snr}.

In Section~\ref{sec:Neural} Fig. \ref{fig:noiselessBCcomparison} and Fig. \ref{fig:noisyfeedbackComparison}, we compare our code with other linear schemes. With perfect feedback, the \gls{lqg} slightly outperforms the \gls{bcl} scheme, with the \gls{eol}, \gls{ol}, \gls{lqg}, and the inner code in the ACLW scheme all exhibiting higher \glspl{bler}. In Fig. \ref{fig:noisyfeedbackComparison}, our proposed scheme outperforms all of the simulated linear codes in the presence of feedback noise.  With a modest increase in blocklength, the \gls{bcl} code can also surpass learned codes under perfect feedback. Nonetheless, the performance of all of the linear codes degrades with feedback noise, motivating the use of \textit{nonlinear} learned feedback codes, which can provide greater robustness and lower \gls{bler} in the \gls{awgnbcf}.

Lastly, one notable difference between the \gls{bcl} scheme and the \gls{lqg} scheme is that the \gls{bcl} scheme uses the first $L$ channel uses to transmit each user's message symbol in an orthogonal TDD-like fashion, where the last $N-L$ channel uses are used for noise cancellation. On the other hand, the \gls{lqg} scheme does not have an apparent TDD initial stage, allowing joint symbol transmission and noise cancellation over all channel uses. Thus, although our proposed scheme utilizes an initial TDD-stage like the \gls{ol}, \gls{eol}, and ACLW schemes, performance could potentially be improved by jointly transmitting each user's message symbol and performing noise cancellation across all $N$ channel uses. This would require designing a time-domain beamforming vector for each user, denoted by $\gvect_\ell\in \R^N$, which determines the power allocation of the $\ell$th user's symbol across channel uses. In fact, we can think of the initial TDD stage of the proposed scheme as a special case of a time-domain beamforming vector. That is, $\gvect_\ell = \evect_1$, as seen in \eqref{eqn:rxSig_analytical}. Nonetheless, our scheme has the same maximum achievable rate as the \gls{lqg} scheme, performs better with feedback noise, and has a ``cleaner" initial transmission protocol, so we leave designing $\gvect_\ell$ as future work.

\begin{figure}[!t]
    \centering    \includegraphics[width=\columnwidth]{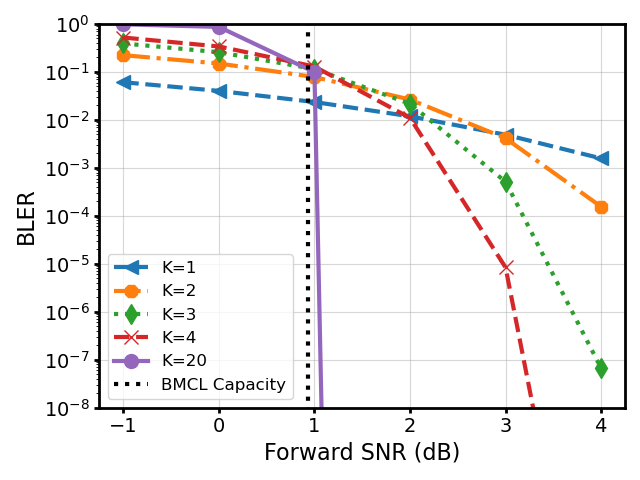}
    \caption{Average \gls{bcl} scheme \gls{bler} per user for various $K_1=K_2=K$ and $N=3K$ with noiseless feedback.}
    \label{fig:OurSchemeSumR23noiseless}
\end{figure}

\begin{figure}[!t]
    \centering
    \includegraphics[width=\columnwidth]{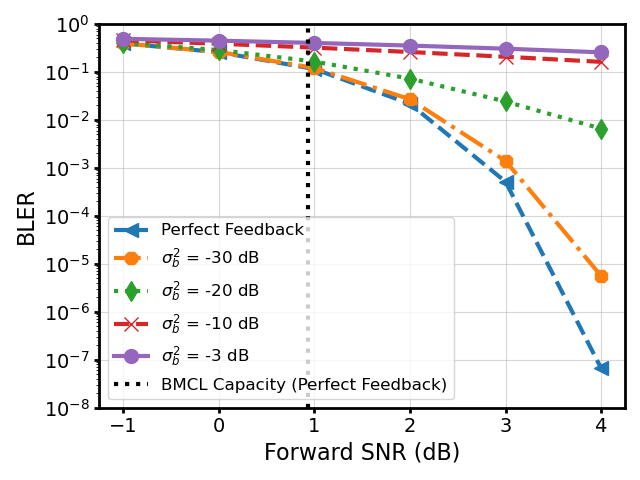}
    \caption{Average \gls{bcl} scheme \gls{bler} per user for $K_1=K_2=3$ and $N=9$ with varying feedback noise powers.}
    \label{fig:OurSchemeSumR23noisy}
\end{figure}





%% file: Neural.tex
In this section, we propose three learned feedback codes for the \gls{awgnbcf} that outperform analytical linear schemes of the same blocklength in terms of BLER, particularly under noisy feedback, and provide initial interpretations of how feedback is utilized for error correction. In addition, we introduce a reduced-complexity extension of lightweight codes inspired by the analytical designs\footnote{Learned codes are available at \url{https://github.com/jacquelinemalayter/DeepBroadcastFeedbackCodes}}. 

\subsection{Lightweight codes for the \gls{awgnbcf}}

We introduce two codes, {\lightBCsep} and {\lightBC}, based on \lightcode~\cite{ankireddy2024lightcode}, each consisting of a \gls{fe} and a \gls{mlp} module, with slight differences in the \gls{fe}, as shown in Fig.~\ref{fig:bclight}. While {\lightcode} uses an input dimension of $d_H = 32$, we set $d_H = 64$ to accommodate the doubled input for the \gls{awgnbcf}, which experimentally yields better BLER performance.

\begin{figure}[!t]
    \centering
    \includegraphics[width=1.1\columnwidth]{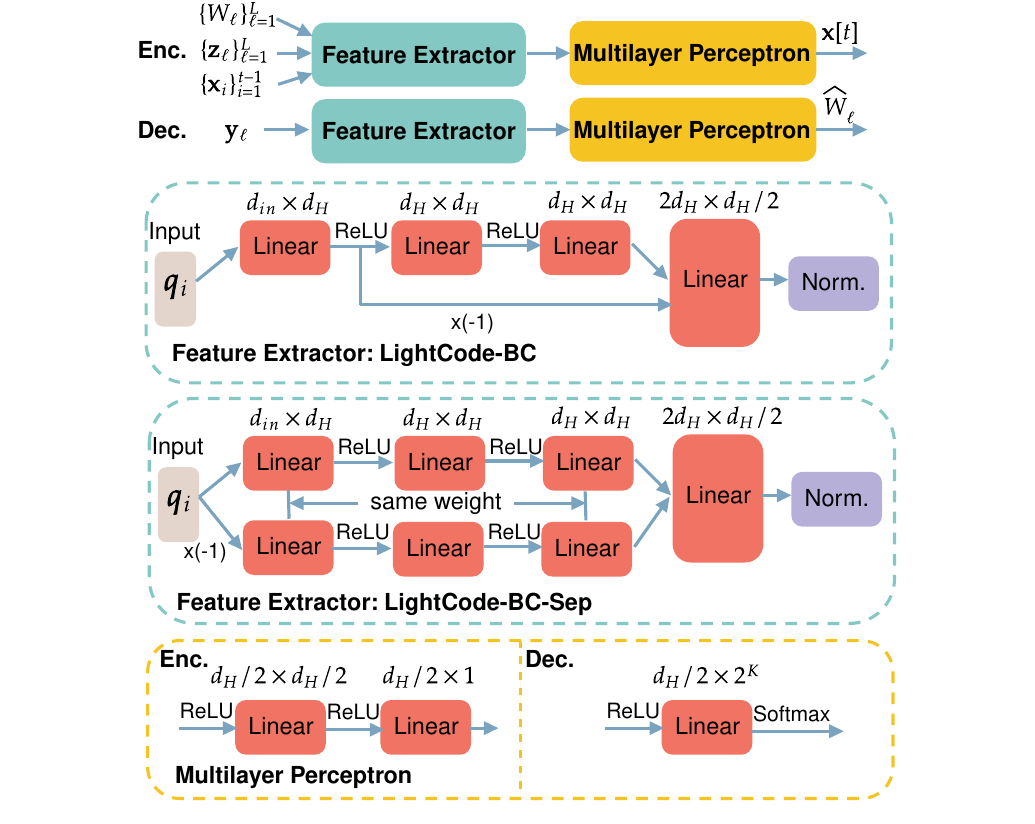}
    \caption{Design (top) and detailed structure (bottom) of the lightweight \gls{awgnbcf} code}
    \label{fig:bclight}
\end{figure}



For encoding, {\lightBCsep}, inspired by the OL scheme~\cite{ozarowAchievable}, transmits the two messages separately in the first two rounds and uses the remaining $N\!-\!2$ rounds for refinement. In contrast, {\lightBC} encodes both messages in the first round and reserves the remaining $N\!-\!1$ rounds for error correction, trading a ``cleaner'' initial message transmission for an additional round of refinement. The parameters highlighted in \textcolor{red}{red} are learnable. Specifically,
\begin{itemize}
    \item {\lightBCsep}: In the first two rounds, each message $W_\ell$ is mapped to a PAM symbol $\Theta_\ell$ and transmitted separately as 
\begin{equation}
    \xvect[t] = \textcolor{red}{\omega_t}\Theta_t, \quad t = 1, 2 \label{eqn:tddround}
\end{equation}
In subsequent rounds ($t = 3, \ldots, N$), the encoder generates symbols based on past transmissions and feedback from both receivers as 
\begin{align}
    \xvect[t] = \textcolor{red}{\omega_t}f_{\textcolor{red}{\alpha_1}}\left(\{\xvect[i],\zvect_1[i], \zvect_2[i]\}_{i=1}^{t-1} \right), 
\end{align}
where $\alpha_1$ denotes the learnable parameters of the encoder. 
    \item {\lightBC}: In the first round, both messages are jointly encoded as 
    \begin{equation}
        \xvect[1] = \textcolor{red}{\omega_1}f_{\textcolor{red}{\alpha_2}}(W_1, W_2),
    \end{equation}
    where $\alpha_2$ denotes the learnable parameters of the encoder. For the remaining rounds ($t=2,\dots,N$), the encoder refines using both messages, past transmissions, and feedback:
    \begin{align}
    \xvect[t] = \textcolor{red}{\omega_t}f_{\textcolor{red}{\alpha_2}}\left(W_1, W_2,\{\xvect[i],\zvect_1[i], \zvect_2[i]\}_{i=1}^{t-1} \right)
\end{align}
\end{itemize}
where $\frac{1}{N}\sum_{t=1}^{N}{\omega_t}^2 = P$ and $\{ \zvect_\ell[i]\}_{i=1}^{t-1}$ represents the feedback from receiver $\ell$. For each receiver, the decoder estimates its message from the received noisy codewords:
\begin{equation}
    \hat{W}_{\ell} =g_{\ell, \textcolor{red}{\phi_{j}}}(\{\yvect_\ell[i]\}_{i=1}^{N}), \quad j = \{1, 2\}
\end{equation}
where $\phi_1$ and $\phi_2$ denote the parameters of {\lightBCsep} and {\lightBC}, respectively.

The \gls{cce} for receiver $\ell$ is defined as
\begin{equation}
    J_{\mathrm{CCE},\ell} = -\frac{1}{B}\sum_{i=1}^B\left(\sum_{j=1}^{C_\ell}p_{\ell,ij}\log(\hat{p}_{\ell,ij})\right)
\end{equation}
where $B$ is the batch size, $C_\ell = 2^{K_\ell}$ is the number of classes (PAM indices) for receiver $\ell$, $p_{\ell, ij}$ is $1$ if class $j$ is the correct label for the $i$-th sample and $0$ otherwise, and $\hat{p}_{\ell, ij} \in \mathbb{R}$ is the predicted probability for that class. The encoder and decoders are jointly trained at the corresponding forward and feedback SNRs to minimize the total loss
\begin{equation}
    J_{\mathrm{light}} = J_{\mathrm{CCE},1} + J_{\mathrm{CCE},2} + \lambda\left(J_{\mathrm{CCE},1} - J_{\mathrm{CCE},2}\right)^2 \label{eqn:lossfcn}
\end{equation} 
where the regularization term can be used to ensure comparable performance across receivers. In \lightBC, $\lambda$ is set to $0$, as each decoders performance is comparable at the end of training and it was experimentally found $\lambda>0$ resulted in inferior \gls{bler}. For training, due to the large dataset size, we generate data stochastically at each iteration instead of relying on a static dataset. That is, we train using per-batch i.i.d. noise realizations with random messages. Training parameters are listed in Table~\ref{tab:training} and model complexity metrics are listed in Table \ref{tab:complexity}.

\begin{table}[!t]
\vspace{6mm}
\caption{Training parameters for each model} \label{tab:training}
\label{params}
\centering
\begin{tabular}{|c|c|c|c|}
\hline
\bfseries Parameters& \bfseries Light-BC-Sep & \bfseries Light-BC &\bfseries RPC-BC\\ \hline
Batch size $B$                     & 100,000 & 100,000 & 50,000\\ \hline
Optimizer                       & AdamW & AdamW & Adam\\ \hline
Weight decay                    & 0.01 & 0.01 & 0 \\ \hline
Epochs                          & 120  & 120 & 120 \\ \hline
Iterations per epoch            & 1000  & 1000 & 2000 \\ \hline
Initial learning rate           & 0.002 & 0.001 & .01\\ \hline
Clipping threshold              & 0.5 & 0.5 & 1 \\ \hline
Power $P$                       & 1 & 1 & 1 \\ \hline
$\lambda$ (regularization)      & 1 & 0 & 0 \\ \hline
\end{tabular}
\end{table}

\subsection{RNN-based codes for the \gls{awgnbcf}}

In addition to the proposed {\lightBCsep} and {\lightBC} codes, we include a \gls{rnn}-based code for comparison. This code is designed to capture temporal correlations in the signal, though at the cost of higher complexity compared to the lightweight codes. We extend RPC~\cite{kim2023robust} to the broadcast setting and denote it by {\rpcBC}, short for robust power-constrained broadcast code. Its structure follows a \textit{state-based} approach, a nonlinear extension of the LQG state-space model used for linear feedback encoding~\cite{kim2023robust}. Specifically, the encoding function is
\begin{align}
    \xvect[t] = \textcolor{red}{\omega_t}f_{\textcolor{red}{\alpha_3}}\paren{\{W_\ell\}_{\ell=1}^2, \{\zvect_\ell[t-1]\}_{\ell=1}^2, \svect[t]} \label{eqn:rnnOutput},
\end{align}
where $\alpha_3$ denotes the learnable encoder parameters, and $\svect \in \R^{(N-1)\times 1}$ is the \textit{state vector}, updated over time by
\begin{align}
    \svect[t] = h_{\textcolor{red}{\gamma}}\paren{\{W_\ell\}_{\ell=1}^2, \{\zvect_\ell[t-1]\}_{\ell=1}^2, \svect[t-1]},
\end{align}
where $\gamma$ denotes the learnable state vector propagation function parameters.

We use a \gls{gru}-based structure (Fig.~\ref{fig:rpcbcenc}) to capture the time correlations in the feedback. Its output is fed into a \gls{mlp}, followed by a power control layer. Unlike RPC~\cite{kim2023robust}, which applies a $\tanh$ activation to the \gls{gru} output, we found experimentally that using an \gls{mlp} with $\mathrm{ReLU}$ activation yields better performance.



\begin{figure}[!t]
    \centering
\includegraphics[width=\columnwidth]{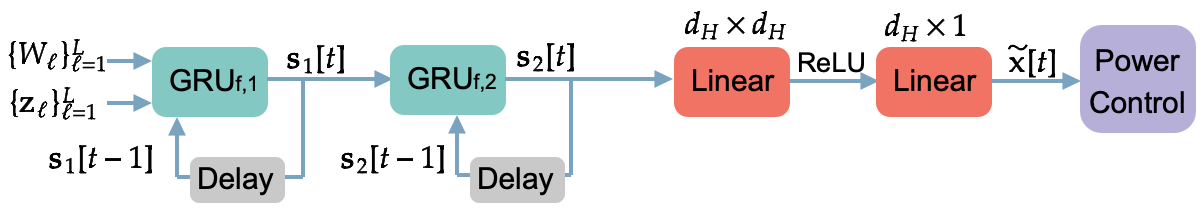}
    \caption{RPC-BC encoder}
    \label{fig:rpcbcenc}
\end{figure}

\begin{figure}[!t]
    \centering
\includegraphics[width=\columnwidth]{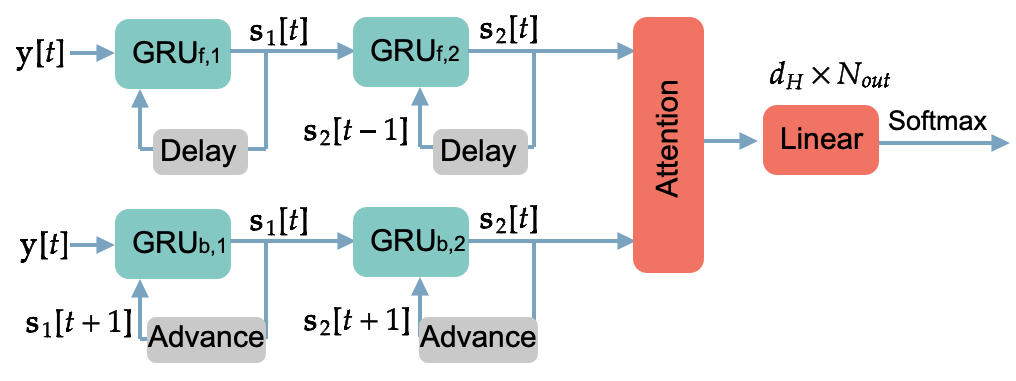}
    \caption{RPC-BC decoder}
    \label{fig:rpcbcdec}
\end{figure}

After $N$ channel uses, the noisy channel outputs are decoded using a structure consisting of a bidirectional \gls{gru}, an attention layer, and an \gls{mlp} as shown in Fig.~\ref{fig:rpcbcdec}. The bidirectional \gls{gru} processes the noisy codewords forward and backward in time, and the attention mechanism then operates on the resulting state vectors to capture long-term dependencies in the received signals~\cite{kim2023robust}. The attention output is passed through an \gls{mlp} with a $\mathrm{softmax}$ activation, whose output represents the ``probability'' of each possible codeword as 
\begin{align*}
    \hat{W}_\ell = g_{\ell,\textcolor{red}{\phi_3}}\paren{\{y_\ell[i]\}_{i=1}^N},
\end{align*}
where $\phi_3$ denotes the learnable decoder parameters.

For training {\rpcBC}, like {\lightBC} and {\lightBCsep}, we generate data stochastically at each iteration instead of relying on a static dataset. Training parameters can be found in Table \ref{tab:training}. Further, model complexity metrics are listed in Table \ref{tab:complexity}. We note that the parameter count of {\rpcBC} is roughly double that of {\lightBC} and {\lightBCsep}, as the decoders in {\rpcBC} have a larger parameter count due to the bi-directional GRU architecture.  The FE in {\lightBCsep} adds slightly more computational complexity per forward encoder/decoder pass as compared to {\lightBC}, which is reflected in the floating point operation (FLOP) count. However, both {\lightBC} and {\lightBCsep} have much fewer FLOPs per forward pass than {\rpcBC} due to the feed-forward model architecture, reinforcing the \textit{lighter-weight} nature of {\lightBC} and {\lightBCsep}.

\begin{table}[!t]
\vspace{6mm}
\caption{Model complexity comparison ($R = 3/9$)} \label{tab:complexity}
\label{params}
\centering
\begin{tabular}{|c|c|c|c|}
\hline
\bfseries Metric & \bfseries Light-BC-Sep & \bfseries Light-BC &\bfseries RPC-BC\\ \hline
Encoder Parameters & 15,201 & 15,649 &14,145\\ \hline
Decoder Parameters & 13,416 & 13,416 &31,054 \\ \hline
Total Parameters & 42,033 & 42,490 & 76,323 \\ \hline
\begin{tabular}{c}
Encoder FLOPs\\ (per forward pass)\end{tabular} & 49,216 & 30,656 & 217,170 \\ \hline
\begin{tabular}{c}
Decoder FLOPs \\ (per forward pass) 
\end{tabular}& 47,616 & 26,240 & 534,188 \\ \hline
\end{tabular}
\end{table}

\subsection{Experimental Results}
The proposed codes are simulated and compared with existing analytical and neural codes, as well as the newly proposed BCL code. In all simulations, it is assumed that the power parameter $P$ is 1. 
\subsubsection{Noiseless Feedback}  
First, we evaluate the proposed codes in the noiseless feedback scenario with symmetric rates, i.e., $K=K_1=K_2$. Fig.~\ref{fig:noiselessBCcomparison} shows their performance, with code rates denoted as $K/N$\footnote{We omit the broadcast version of DeepCode in~\cite{li2022deep}, since its code rate of $1/3$ yields worse BLER performance than the OL scheme with code rate $3/9$.}. The \gls{lqg} and \gls{bcl} scheme perform better in terms of BLER than the other simulated linear schemes, with the \gls{lqg} scheme slightly outperforming the \gls{bcl} scheme at high SNR. For \gls{bcl}, longer messages with more channel uses lead to a rapid decrease in \gls{bler} at high SNRs, while shorter messages with larger \gls{pam} spacing perform better at low SNRs, a trend also observed in other analytical and learned codes. For $K=1$ and $N =3$, {\lightBC} outperforms {\lightBCsep}, showing that joint transmission with additional error correction is more effective. At $K=3$ and $N=9$, {\lightBC} and {\lightBCsep} perform similarly, while {\rpcBC}, benefiting from noise averaging, achieves slightly lower \gls{bler} at low SNRs. For the same blocklength, learned codes outperform analytical linear codes. However, the proposed \gls{bcl} scheme can match or surpass their performance by roughly doubling the blocklength $N$ as seen from the $R=7/21$ BCL curve in Fig. \ref{fig:noiselessBCcomparison}. The additional latency is negligible compared to their computational simplicity.

\begin{figure}[!t]
    \centering
\includegraphics[width=\columnwidth]{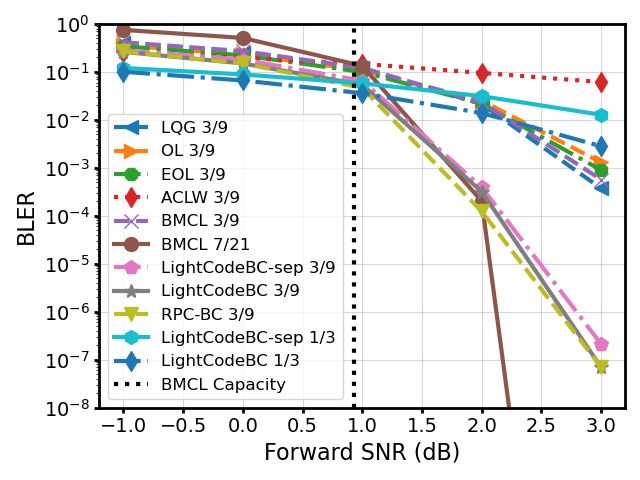}
    \caption{Broadcast code comparison for $L=2$ users with perfect feedback.}
    \label{fig:noiselessBCcomparison}
\end{figure}

\subsubsection{Noisy Feedback} 
Next, we compare the codes under varying feedback noise at a fixed forward \gls{snr} of 4 dB. As shown in Fig.~\ref{fig:noisyfeedbackComparison}, the learned codes outperform the analytical codes when feedback noise is present. The blocklength advantage of the analytical code disappears, and the $R=7/21$ BCL code suffers from noise accumulation. Even with small feedback noise, learned codes demonstrate a \gls{bler} advantage over simpler analytical schemes. Furthermore, the {\lightBC} scheme performs best in the blocklength regime of interest, demonstrating that RNN-based schemes with more parameters and complexity do not necessarily provide performance advantages.

\begin{figure}[!t]
    \centering
    \includegraphics[width=\columnwidth]{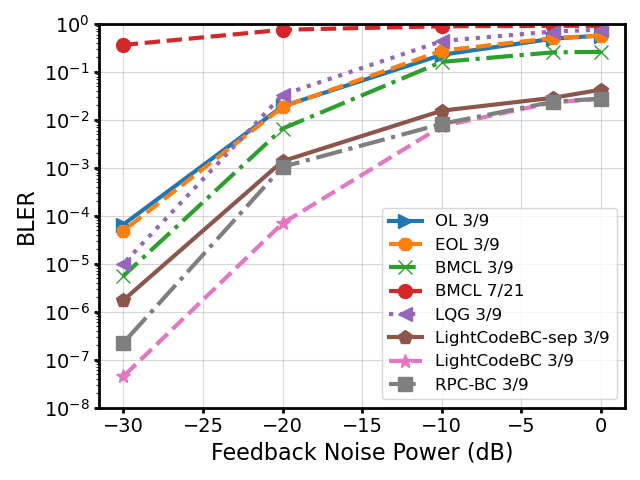}
    \caption{ 
    Broadcast code comparison for $L=2$ users at forward \gls{snr} of 4 dB with different feedback noise powers.
    }
    \label{fig:noisyfeedbackComparison}
\end{figure}

\subsubsection{TDD Comparison} 
 \begin{figure}[!t]
    \centering
    \includegraphics[width=\columnwidth]{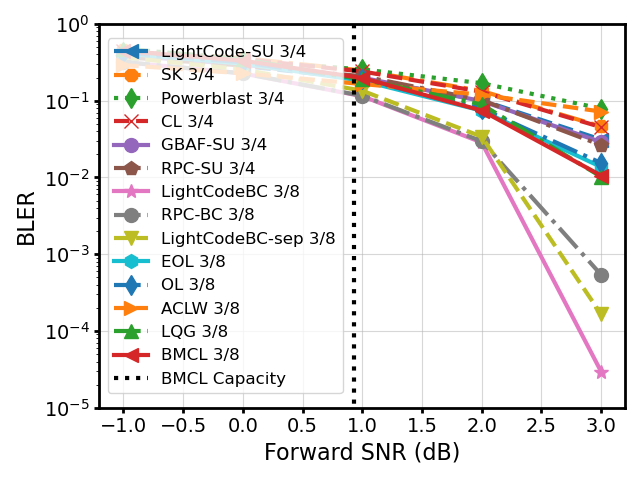}
    \caption{
     Broadcast codes versus TDD single-user feedback codes with perfect feedback ($K=3$ per user, $N=4$ TDD / $N=8$ broadcast)
    }
    \label{fig:noiselessTDDcomparison}
\end{figure}

 \begin{figure}[!t]
     \centering
     \includegraphics[width=\linewidth]{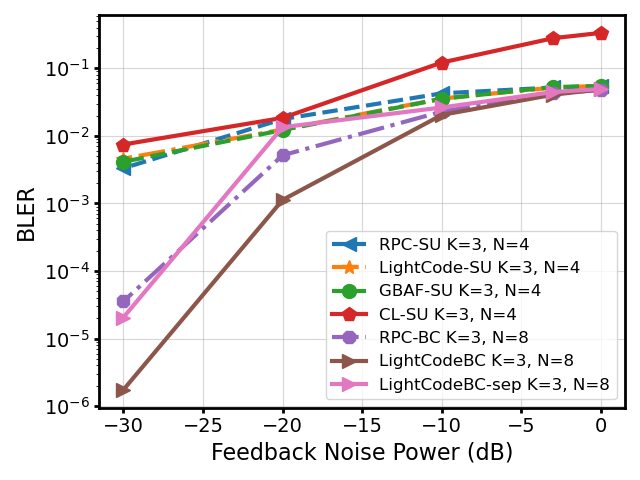}
     \caption{
     Broadcast codes versus TDD single-user feedback codes with noisy feedback
     }
     \label{fig:baselines}
 \end{figure}


Finally, we compare the learned feedback codes with baseline \gls{tdd} codes. By this, we mean that the number of message bits per user $K$ is kept the same, but the number of channel uses per user is halved, examining the gain from sending messages for two users jointly versus in orthogonal channels. Specifically, for the broadcast codes $K=3$ per user with $N=8$, while in the TDD case $K=3$ per user with $N=4$. Fig.~\ref{fig:noiselessTDDcomparison} presents the noiseless feedback results, while Fig.~\ref{fig:baselines} shows the noisy feedback results at a fixed forward SNR of 4 dB, compared against \gls{tdd} \gls{su} codes: SK~\cite{schalkwijk1966coding}, Powerblast~\cite{ankireddy2024lightcode} (a variation of SK that transmits the estimated PAM index difference in the last round instead of the real-valued estimation error), CL~\cite{chance2011concatenated}, {\lightcode}~\cite{ankireddy2024lightcode}, RPC~\cite{kim2023robust}, and GBAF~\cite{ozfatura2022all}. The results show that, at the same code rate, both analytical and learned broadcast codes outperform \gls{tdd} \gls{su} codes, demonstrating that cooperation between the encoder and decoders effectively exploits feedback from both users. At the demonstrated blocklengths, {\lightBC} performs the best at 3 dB forward \gls{snr}, demonstrating that the additional round of noise refinement over {\lightBCsep} is beneficial.

\subsubsection{Loss Convergence Behavior}
In addition to plotting the \gls{bler} performance of the proposed codes, we observe the behavior of the per user training loss function convergence across training epochs for {\lightBC}, {\lightBCsep}, and {\rpcBC}. We plot the loss for $R_1=R_2 = 3/9$, a broadcast chanel SNR of $4$ dB, and feedback noise power of $-20$ and $-30$ dB in Figs. \ref{fig:20dbLoss} and \ref{fig:30dbloss}, respectively. Empirically, it was found in {\lightBC} and {\rpcBC} that the addition of a regularization term in the loss function in \eqref{eqn:lossfcn} reduced the BLER performance. It was also empirically found that the absence of a regularization term in {\lightBCsep} resulted in one user performing very well, while the other performed poorly. This may be because {\lightBCsep} sends each users' symbols separately, while {\lightBC} and {\rpcBC} jointly map each users' symbols in the first channel use. Nonetheless, {\lightBCsep} displays more equal loss behavior during training than {\lightBC} and {\rpcBC}, likely a direct result of regularization in \eqref{eqn:lossfcn}. 

\begin{figure}
    \centering
    \includegraphics[width=.9\columnwidth]{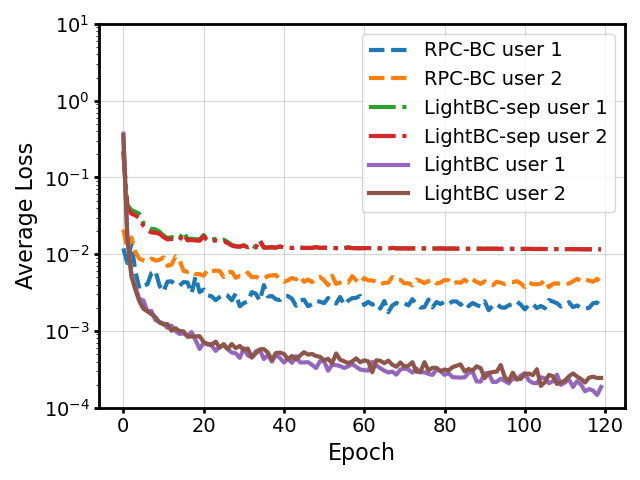}
    \caption{Per user loss function convergence as a function of training epochs for {\lightBC}, {\lightBCsep}, and {\rpcBC} for $R_1=R_2=3/9$, broadcast channel SNR of $4$ dB, and feedback noise power of $-20$ dB.}
    \label{fig:20dbLoss}
\end{figure}

\begin{figure}
    \centering
    \includegraphics[width=0.9\linewidth]{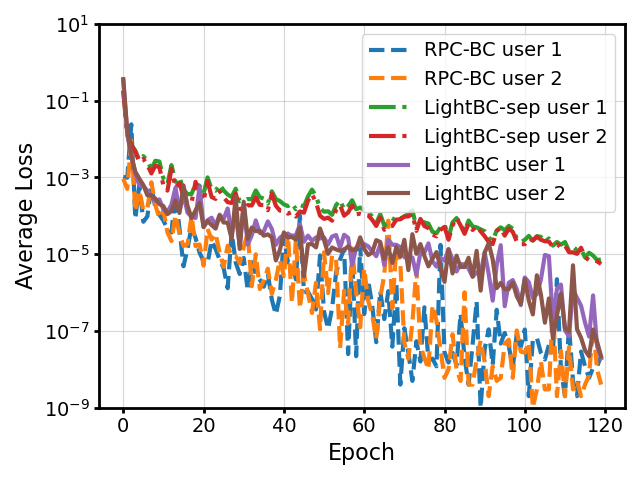}
     \caption{Per user loss function convergence as a function of training epochs for {\lightBC}, {\lightBCsep}, and {\rpcBC} for $R_1=R_2=3/9$, broadcast channel SNR of $4$ dB, and feedback noise power of $-30$ dB.}
    \label{fig:30dbloss}
\end{figure}

\subsection{Interpretation of Neural Codes}
To understand how feedback contributes to error correction, we provide an initial interpretation of the learned broadcast codes. Scatter-plot visualization of the encoder outputs reveals notable similarities between {\rpcBC} and {\lightBC}. Moreover, all learned codes display a power-efficient structure, transmitting only when \textit{necessary}.


We consider the simple case with $K_1= K_2=1$, $N = 3$,  forward SNR of 3 dB, and noiseless feedback. Each figure is generated from 1,000 data points. Fig.~\ref{fig:jackie_x1} shows the transmitted symbol at round~1 for the learned broadcast codes. {\lightBCsep} sends the two messages separately in the first two rounds using standard BPSK modulation. In contrast, {\rpcBC} and {\lightBC} transmit both messages together $\xvect[1] = f_{\theta}(W_1, W_2)$ in a single round. Unlike \gls{pam}, the resulting symbols are not equally spaced but instead encode whether $W_1$ and $W_2$ are the same or different. For example, in {\rpcBC}, $W_1= W_2=1$ (red) results in a positive value, while $W_1 = W_2=0$ (blue) yields a negative value. When $W_1 \neq W_2$, the symbols are close to zero and hard to distinguish. {\lightBC} exhibits the opposite pattern, mapping differing messages to large amplitudes and identical messages to near-zero values.

\begin{figure}
    \centering
    \includegraphics[width=\columnwidth]{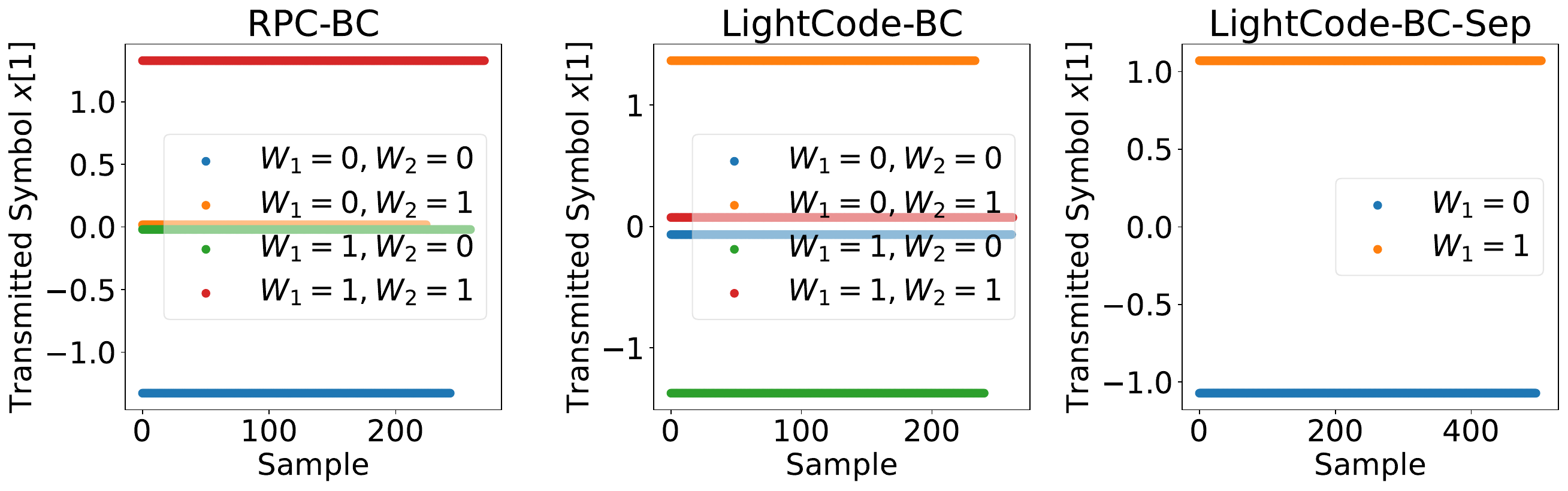}
    \caption{Encoder output $\xvect[1]$ of the learned broadcast codes.}
    \label{fig:jackie_x1}
\end{figure}

Next, we analyze the second-round transmitted symbol $\xvect[2]$ in {\rpcBC} and {\lightBC} with respect to the first-round forward noise. In Fig.~\ref{fig:jackie_x2_n}, we plot $\xvect[2]$ against the forward noise of user~1 ($\nvect_1^b[1]$) assuming $\nvect_2^b[1] = 0$, and vice versa. The second-round transmission serves two main functions. Using {\rpcBC} as an example:
\begin{itemize}
    \item Power-efficient: When $W_1=0$ and $\nvect_1^b[1] < 0$, binary detection decodes the message correctly without additional information, so the transmitted symbol is around $0$. If $\nvect_1^b[1] > 0$, a decoding error may occur and $\xvect[2]$ transmits a scaled version of the noise. The same behavior holds for $W_1 = 1$, and similarly for $W_2$, but with opposite sign. This nonlinear, power-efficient structure resembles a ``ReLU-like'' shape, avoiding unnecessary transmissions, which is also observed in the learned \gls{su} codes~\cite{zhou2024isit}.
   \item Differentiable: When $W_1 \neq W_2$ (indistinguishable in round~1), the transmitted symbol $\xvect[2]$ in round~2 shifts to a large positive (orange) or negative (green) amplitude, enabling differentiation between $W_1=0, W_2=1$ and $W_1=1, W_2=0$.
\end{itemize}
A similar behavior is observed in {\lightBC}.

\begin{figure}
    \centering
    \includegraphics[width=\columnwidth]{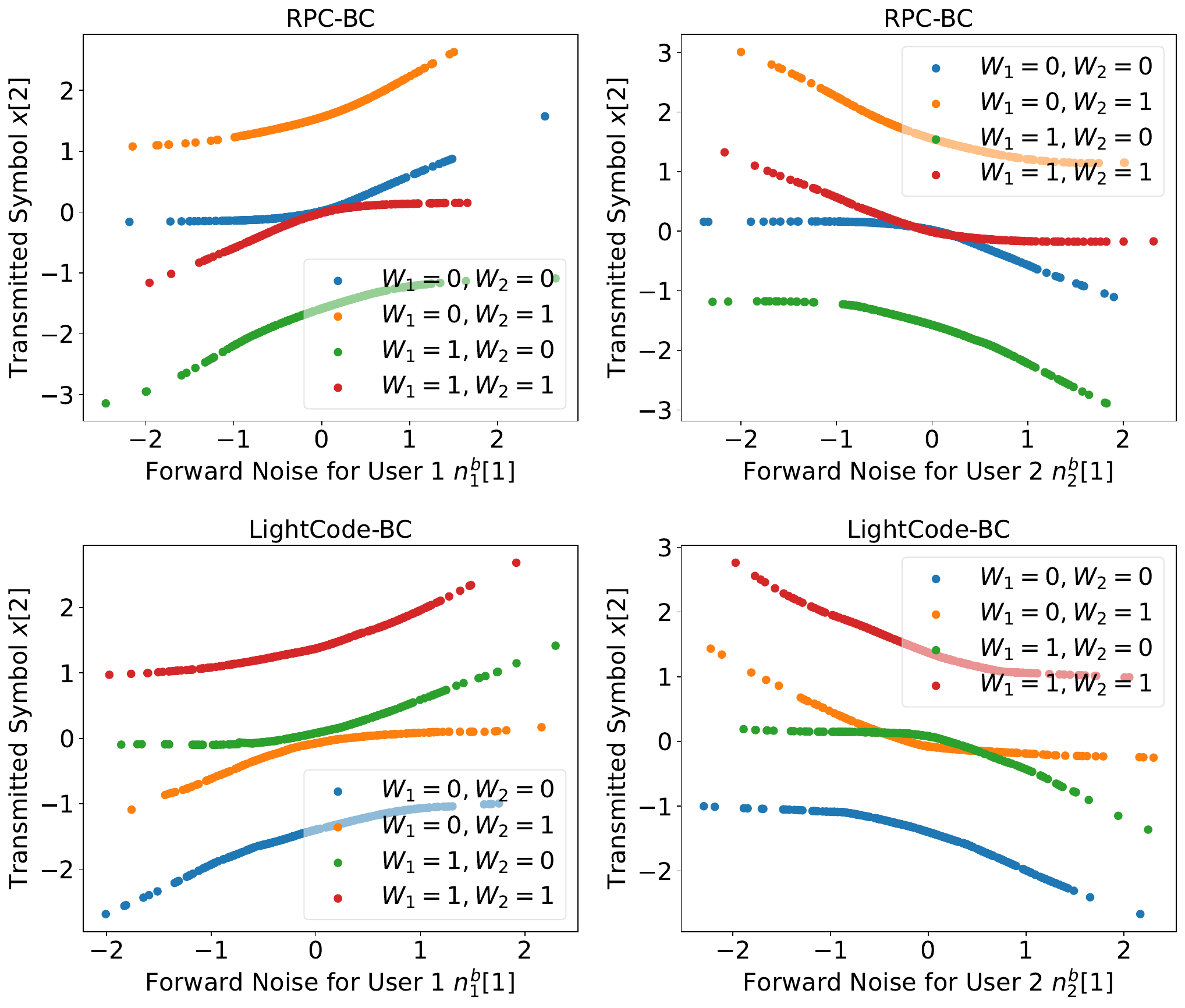}
    \caption{Encoder output $\xvect[2]$ versus forward noise during the first round for different users}
    \label{fig:jackie_x2_n}
\end{figure}

Finally, Fig.~\ref{fig:bc_sep_x3} shows the third-round symbol $\xvect[3]$ versus the forward noises. As before, the other user's noise is fixed at zero for clarity. {\lightBCsep} exhibits the same ReLU-like, power-efficient behavior as in earlier rounds. For {\lightBC} and {\rpcBC}, the behavior is more nonlinear. When the noise is favorable and aids decoding, no additional information is transmitted. For example, in {\rpcBC}, if $\xvect[2]$ is positive (yellow and blue) and the second-round noise $n_1^b[2]$ is also positive, the transmitted symbol $\xvect[3]$ remains close to zero. Only when the negative noise corrupts the message does $\xvect[3]$ actively perform error correction. 


\begin{figure}
    \centering
    \includegraphics[width=\linewidth]{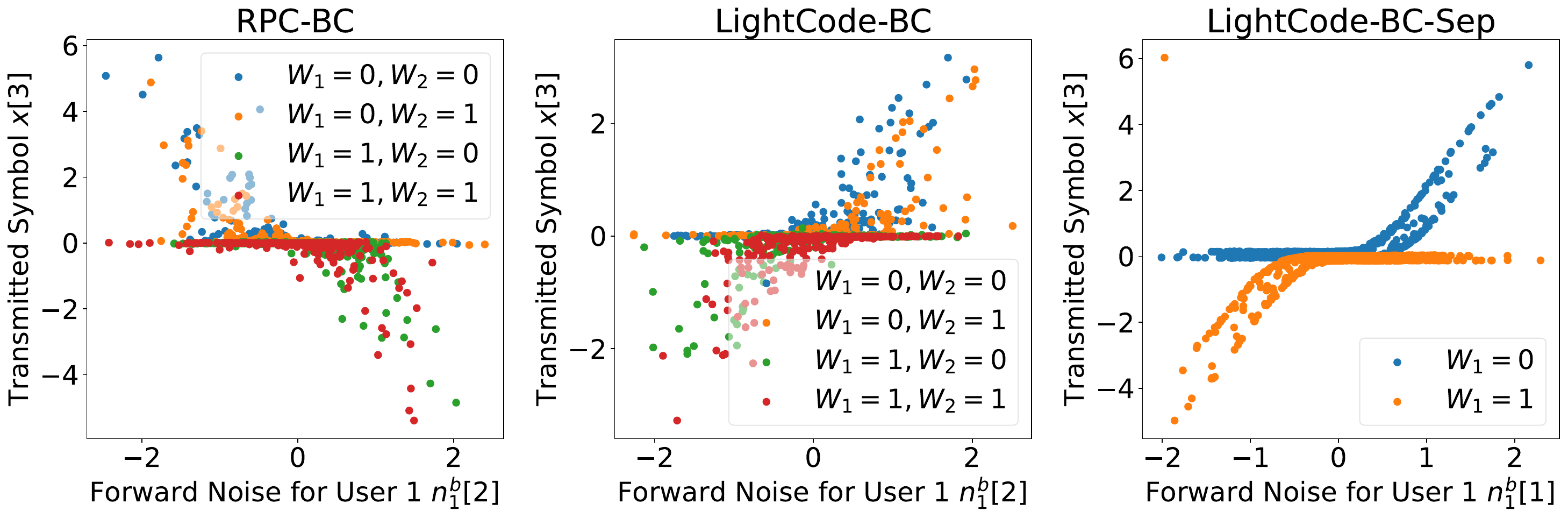}
    \caption{
    Encoder output $\xvect[3]$ versus second-round forward noise $n_1^b[2]$ for {\lightBC} and {\rpcBC} (first-round forward noise $n_1^b[1]$ for {\lightBCsep})
    }
    \label{fig:bc_sep_x3}
\end{figure}

 \subsection{A lighter-weight training scheme for the symmetric rate region}

A limitation of the learned schemes is that each decoder is assumed to be unique during training, requiring joint training of an encoder with two separate decoders. This makes training computationally intensive and difficult to generalize to $L>2$ users.  To address these challenges, we propose a lighter-weight training scheme for learned broadcast codes, inspired by the spreading-code structure used in some analytical codes.

 Consider the \gls{bcl} scheme in Section~\ref{sec:proposedCode}, which uses a spreading-code like structure in the encoder and linear combiner. The decoded output is 
 \begin{align*}
     \hat{\Theta}_\ell = \qvect^T\Cmat_\ell \yvect_\ell,
 \end{align*}
 where $\qvect$ is a linear combiner that is the same between users. The major difference between users in this scheme is the spreading code $\Cmat_\ell$, where the underlying encoding matrix $\Fmat$ is the same due to rate symmetry. With this in mind, we propose modifying the scheme as shown in Fig. \ref{fig:modifiedScheme}, by multiplying the input to each decoder by a unique spreading code matrix $\Cmat_\ell$ and sharing the weights of each decoder. 
 \begin{figure}[!t]
     \centering
     \includegraphics[width=.8\linewidth]{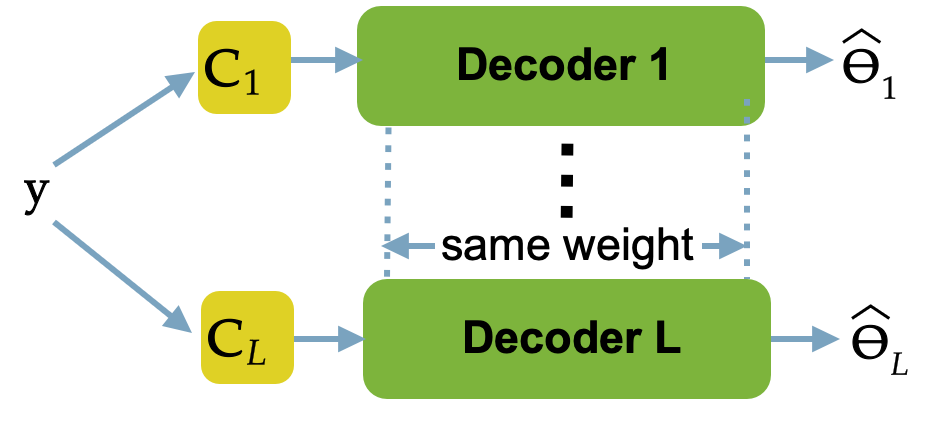}
     \caption{Modified lightweight model for the symmetric rate region}
     \label{fig:modifiedScheme}
 \end{figure}



\begin{figure}[!t]
    \centering
\includegraphics[width=0.9\columnwidth]{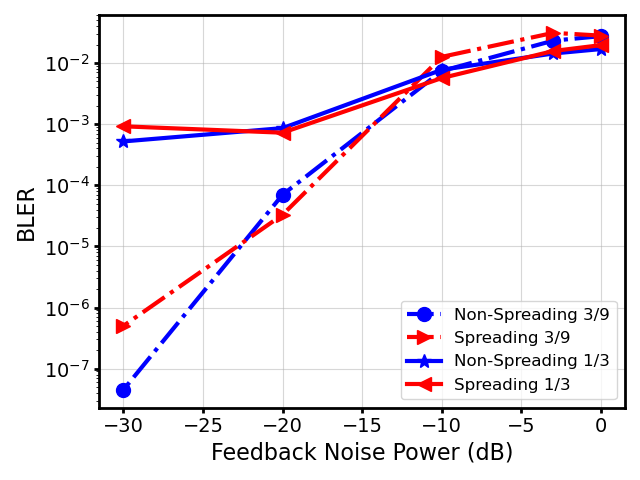}
    \caption{Performance of {\lightBC}  versus {\lightBC}   implemented with spreading codes.}
    \label{fig:LightBCSpread2users}
\end{figure}

\begin{figure}[!t]
    \centering
\includegraphics[width=0.9\columnwidth]{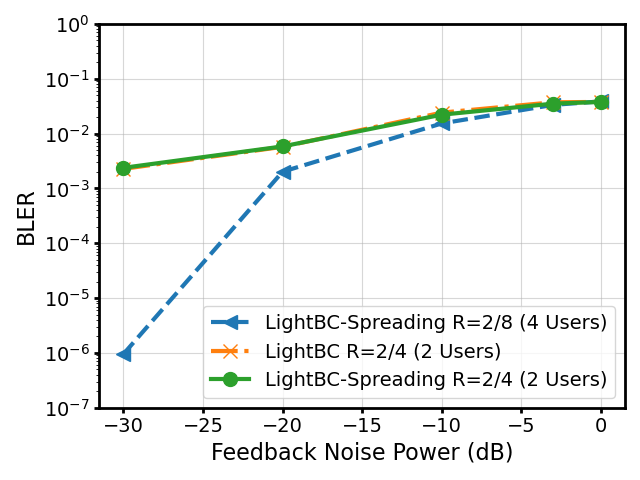}
    \caption{Demonstration of performance gain from joint encoding over 4 users.}
    \label{fig:LightBCSpread4users}
\end{figure}

We train the code using the same parameters as {\lightBC} in Table \ref{tab:training}. With noisy feedback, we see that the performance difference between the lightweight scheme and the original \lightBC \ is mild, or in some instances, even better than the original \lightBC, where the largest performance difference is observed at $-30$ dB feedback noise power. 

Finally, we train this scheme on 4 users. We train with  forward SNR 6 dB and code rate $R = 2/8$ with various feedback noises. We also compare with both \lightBC \ and \lightBC \ implemented with spreading codes for $2$ users and rate $R = 2/4$. The proposed lightweight version for $4$ users outperforms in terms of BLER by over 3 orders of magnitude. This demonstrates that the lighter-weight scheme in the symmetric rate region with shared decoder weights allows more efficient training for $L>2$ users by utilizing insights from analytical coding schemes.  





%% file: Conclusion.tex
In this paper, we reviewed existing codes for the \gls{awgnbcf} and introduced a new analytical code. We also extended learned codes for the single user \gls{awgnf} to the broadcast channel, and demonstrated these codes could provide more reliable performance with noisy feedback and provided corresponding interpretation. Finally, we proposed a lighter-weight training scheme based on insights from analytical codes. This work demonstrates the utility of learned channel output feedback codes in practical noisy feedback scenarios and offers numerous future research directions. For example, while the proposed schemes assumed passive feedback, power or rate constrained devices may benefit from active and limited feedback schemes. Additionally, our proposed lightweight structure is applicable to federated training schemes, as shared decoding weights would limit the number of parameters passed between devices in training, which could allow learning in real wireless environments. 

%% file: bclmaxsumrate.tex
\setcounter{lemma}{0}
\setcounter{remark}{0}

\begin{lemma}The asymptotic power of the encoding matrix $\Fmat$ \eqref{eqn:matrixPowerLimit} is strictly decreasing with $\beta\in(0,1)$.
\end{lemma}
\begin{proof}
    From \eqref{eqn:matrixPowerLimit}, let $f(\beta) = \frac{(1-\beta^{2L})^2}{L^2\beta^{2L}(1-\beta^2)}$ defined on the interval $\beta\in(0,1)$. Then the derivative of $f(\beta)$ is 
    \ifdraft{\begin{align*}
    f'(\beta) &= \frac{-1}{c(\beta)}\bigg(4L\beta^{2L-1}(1 - \beta^{2L})(L^2 \beta^{2L}(1 - \beta^2)) +  (1-\beta^{2L})^2(2L^3 \beta^{2L-1}(1 - \beta^2) - 2L^2 \beta^{2L+1})\bigg),   
    \end{align*}}{
    \begin{align*}
    f'(\beta) &= \frac{-1}{c(\beta)}\bigg(4L\beta^{2L-1}(1 - \beta^{2L})(L^2 \beta^{2L}(1 - \beta^2)) + \\& (1-\beta^{2L})^2(2L^3 \beta^{2L-1}(1 - \beta^2) - 2L^2 \beta^{2L+1})\bigg),   
    \end{align*}}
    where $c(\beta) = {\paren{L^2 \beta^{2L} (1 - \beta^2)}^2}$, which is positive. Thus it can be shown $f'(\beta)<0$ by showing  
    \ifdraft{\begin{align*}
        4L\beta^{2L-1}(1 - \beta^{2L})(L^2 \beta^{2L}(1 - \beta^2))+(1-\beta^{2L})^2(2L^3 \beta^{2L-1}(1 - \beta^2)&> (1-\beta^{2L})^2 2L^2 \beta^{2L+1}, 
    \end{align*} }{
    \begin{align*}
        4L\beta^{2L-1}(1 - \beta^{2L})(L^2 \beta^{2L}(1 - \beta^2))&+\\(1-\beta^{2L})^2(2L^3 \beta^{2L-1}(1 - \beta^2)&> (1-\beta^{2L})^2 2L^2 \beta^{2L+1}, 
    \end{align*} }
    which is implied by proving
    \begin{align*}
        L(1-\beta^2) > (1-\beta^{2L}).
    \end{align*}
    The above inequality is valid by the mean value theorem.

\end{proof}

\begin{lemma}
\label{lemma:capacity}
    The $L$-user \gls{bcl} maximum sum rate for an SNR of $\frac{P}{\sigma_b^2}$ and perfect feedback is given by $$C_{sum}\paren{\frac{P}{\sigma_b^2}} = -L\log
    _2\paren{\beta_{\infty}},$$ where $$\beta_\infty = \paren{\beta : \frac{(1-\beta^{2L})^2}{L^2\beta^{2L}(1-\beta^2)} = \frac{P}{\sigma_b^2 L}}.$$ 
\end{lemma}
\begin{proof}
Let $R_{sum} = LR_\ell$ where $R_\ell$ is the rate of each user and define $\beta(\gamma)\in(0,1)$ as 
\begin{align}
\beta(\gamma) = \paren{\beta:\frac{(1-\beta^{2L})^2}{L^2\beta^{2L}(1-\beta^2)} = \frac{P\gamma}{L(\sigma_b^2)}}
\end{align}
 Fix a $\gamma\in(0,1)$. From the average \gls{bler} in \eqref{eqn:newbler} and using the noise power in \eqref{eqn:NPlargeN}, the term in the $Q$-function for large $N$ can be written as approximately
    \begin{align}
       \sqrt{c_2N\frac{\beta(\gamma)^{-2N}}{2^{2NR_{\ell}}}}, \label{eqn:InQfunc}
    \end{align}
    where $c_2$ is a positive constant. Then, 
    \begin{align}
        \lim_{N\to\infty}\sqrt{c_2N\frac{\beta(\gamma)^{-2N}}{2^{2NR_{\ell}}}}  = \begin{cases}
            0 &  R_\ell > -\log_2\paren{\beta(\gamma)}\\
            \infty & R_\ell \leq -\log_2\paren{\beta(\gamma)}
        \end{cases}
    \end{align}
    implying that for any $R_\ell\leq -\log_2\paren{\beta(\gamma)}$, $\mathbb{P}_{e,\ell}\to0$ as $N\to\infty$.  Thus, the maximum sum rate is given as
    \begin{align*}
        C_{sum}\paren{\frac{P}{\sigma_b^2}} = \sup_{\beta(\gamma)}-L\log_2\paren{\beta(\gamma)}.
    \end{align*}
    Since \eqref{eqn:matrixPowerLimit} is strictly increasing as $\beta$ decreases, we let $\gamma\to1$ so that 
    \begin{align*}
         C_{sum}\paren{\frac{P}{\sigma_b^2}}  = -L\log_2\paren{\beta(1)}.
    \end{align*}
    Finally, define $\beta_\infty = \beta(1)$ and the proof is complete. 
\end{proof}

\begin{remark}
    The \gls{bcl} maximum sum rate for an SNR of $\frac{P}{\sigma_b^2}$ and perfect feedback has a finite limit as the number of users $L$ goes to infinity. That is, 
\begin{align*}
    \underset{L\to\infty}{\lim}C_{sum}\paren{\frac{P}{\sigma_b^2}} = \frac{\alpha}{\ln2} 
\end{align*}
where $\alpha$ satisfies
\begin{align*}
    \frac{(1-e^{-2\alpha})^2}{2\alpha e^{-2\alpha}} = \frac{P}{\sigma_b^2}.
\end{align*}
\end{remark}

\begin{proof}
    We first claim $\beta_\infty\to 1$ as $L\to \infty$. For each $L \in \mathbb{N}$, let $\beta_L \in (0,1)$ denote the unique solution of
\begin{equation}
    \frac{(1-\beta_L^{2L})^2}{L^2 \beta_L^{2L} (1-\beta_L^2)}
    = \frac{P}{\sigma_b^2 L}. 
    \label{eq:beta_L_def}
\end{equation}
Suppose by contradiction, $\beta_L = (1-\varepsilon)$ for any $0<\varepsilon<1$. The left side of \eqref{eq:beta_L_def} has the limit
    \begin{align*}
        \lim_{L\to\infty}\frac{(1-\beta_L^{2L})^2}{L^2 \beta_L^{2L} (1-\beta_L^2)} = \infty,
    \end{align*} 
    but 
    \begin{align*}
        \lim_{L\to\infty} \frac{P}{L\sigma_b^2} = 0.
    \end{align*}
    By contradiction, the claim $\beta_\infty\to 1$ as $L\to \infty$ holds.
    
    We therefore parameterize $\beta_L = (1-\frac{\alpha}{L})$ for some $\alpha>0$. Substituting $\beta_L$ into \eqref{eq:beta_L_def} and letting $L$ grow large, the equation becomes
    \begin{align}
        \frac{(1-e^{-2\alpha})^2}{2\alpha e^{-2\alpha}} = \frac{P}{\sigma_b^2}. \label{eq:alphadef}
    \end{align}
    Likewise, substituting $\beta_L$ into the capacity limit in Lemma \ref{lemma:capacity}, it follows
    \begin{align*}
        C_{sum}\!\left(\frac{P}{\sigma_b^2}\right)
    = -L \log_2(\beta_L)
    = -\frac{L}{\ln 2} \ln\!\left(1 - \frac{\alpha}{L}\right). 
    \end{align*}
    Using the expansion $\ln(1-x) = -x + O(x^2)$ as $x \to 0$ with $x = \alpha/L$, we get
\[
    -L \ln\!\left(1 - \frac{\alpha}{L}\right)
    = \alpha + O\!\left(\frac{1}{L}\right),
\]
and therefore
\[
    \lim_{L \to \infty} C_{sum}\!\left(\frac{P}{\sigma_b^2}\right)
    = \frac{\alpha}{\ln 2},
\]
where $\alpha$ satisfies \eqref{eq:alphadef}. 
\end{proof}